\documentclass[12pt]{article}
\usepackage[utf8]{inputenc}

\usepackage{latexsym,amssymb,amsfonts,amsmath,amsthm}
\usepackage[dvips]{graphicx}
\usepackage{comment}

\newtheorem{theorem}{Theorem}
\newtheorem{lemma}{Lemma}
\newtheorem{corollary}{Corollary}
\newtheorem{proposition}{Proposition}
\newtheorem{remark}{Remark}
\newtheorem{definition}{Definition}

\makeatletter
\newtheorem*{rep@theorem}{\rep@title}
\newcommand{\newreptheorem}[2]{%
\newenvironment{rep#1}[1]{%
 \def\rep@title{#2 \ref{##1}}%
 \begin{rep@theorem}}%
 {\end{rep@theorem}}}
\makeatother
\newreptheorem{theorem}{Theorem}

\numberwithin{theorem}{section}
\numberwithin{lemma}{section}
\numberwithin{corollary}{section}
\numberwithin{proposition}{section}
\numberwithin{remark}{section}
\numberwithin{definition}{section}

\newcommand{\bs}[1]{\boldsymbol{#1}}

\newcommand{\supp}{{\rm supp}}

\newcommand{\dist}{{\rm dist}}

\usepackage{color}

\usepackage{ulem}

\title{Circuit equation of Grover walk}
\author{Yusuke Higuchi$^1$ and Etsuo Segawa$^2$\\
$^1${\small Department of Mathematics,
Gakushuin University,} \\
{\small Tokyo 171-8588, Japan.}
\\
$^2${\small Graduate School of Environment and Information Sciences, Yokohama National University,}\\ {\small Yokohama 240-8501, Japan.
}
}
\date{}

\begin{document}

\maketitle

\par\noindent
{\bf Abstract}. 
We consider the Grover walk on the infinite graph
in which an internal finite subgraph receives the inflow from the outside with some frequency 
and also radiates the outflow to the outside. 
To characterize the stationary state of this system, which is represented by a function on the arcs of the graph, 
we introduce a kind of discrete gradient operator twisted by the frequency. 
Then we obtain a circuit equation which shows that (i) the stationary state is described by the twisted gradient of a potential function which is a function on the vertices; (ii) the potential function satisfies the Poisson equation with respect to a generalized Laplacian matrix. 
Consequently, we characterize the scattering on the surface of the internal graph and the energy penetrating inside it.
Moreover, for the complete graph as the internal graph, we illustrate the relationship of the scattering
and the internal energy to the frequency and the number of tails. \\



\section{Introduction}

Let us consider a situation in which water is poured continuously with an appropriate constant flow rate from a faucet to a tank with drains on the bottom. Over a sufficiently long period of time, the amount of water in the tank will become constant and stationary because of the balance between the inflow from the faucet and the outflow through the drains. In this paper, we consider a quantum walk model with a similar setup. We set a connected and finite graph and choose some vertices of this graph as the boundaries. 
From these boundaries, we feed a quantum walker with some frequency at every time step. Quantum walk has appeared as a dynamics that plays an important role in achieving  quadratical speed up in the quantum searches on graphs~\cite{Ambainis2003, Childs, Portugal}. 
Most of studies concerning the quantum walk  i.e,~\cite{Ambainis2003, Childs, Portugal} have been 
developed in the $\ell^2$ space and have provided interesting results, but we are also interested in the behavior of the quantum walk in a more extended space such as the  $\ell^\infty$ space. As a result, we will naturally obtain a {\it stationary} state of the quantum walk as a dynamical system. 
To this end, we set the following graph, time evolution and initial state. 
First, we connect the semi-infinite paths, named the tails, to the boundary vertices of the original graph $G_0=(V_0,A_0)$. The resulting semi-infinite graph is denoted by $\tilde{G}=(\tilde{V},\tilde{A})$. 
Here $A_0$ and $\tilde{A}$ are the sets of the symmetric arcs of $G_0$ and $\tilde{G}$, respectively. 
Secondly, we set the time evolution such that the local time evolution at each vertex follows the Grover matrix. Note that, in the region ahead of the boundary, the time evolution is free. 
Finally, as the initial state, we set a uniformly bounded initial state on the tails so that a quantum walker penetrates the internal graph at every time step. 
See Figure~\ref{Fig:graphs} for the complete graph with $4$ vertices. 
This quantum walk model was first formalized in \cite{FelHil1,FelHil2} on general graphs.  
Such a model can be interpreted as a discrete-analogue of the stationary Schr{\"o}dinger equation~\cite{Albe}. 
For example, relations of discrete-time quantum walks on the one-dimensional lattice to 
the Schr{\"o}dinger equation with delta potentials and also continuous potential in the real line can be seen in ~\cite{MMOS,KentaHiguchi,Morioka}. 
The notion of the stationarity of quantum walks has been considered in \cite{Konno} by extending the total state space $\ell^2(\mathbb{Z};\mathbb{C}^2)$ to the vector space $(\mathbb{C}^2)^{\mathbb{Z}}$. 

Let us state related works~\cite{HSS1,HSS2} on the Grover walks in the above setup from the view point of the frequency of the inflow parameterized by $z$ or $\xi$. 
Details of the definitions and settings can be seen in Section~2.
The $n$-th iteration of the quantum walk $\psi_n\in \mathbb{C}^{A_0}$ itself does not converge by the frequency $z\in \partial\mathbb{D}=\{ w\in \mathbb{C} \;:\; |w|=1\}$ of the inflow.  However it is shown in \cite{FelHil1,FelHil2,HS} that this dynamics is attracted to a stable orbit which is described by $\mathcal{A}_z=\{z^{-\ell}\phi_z: \ell\in \mathbb{N}\}$. Here $\phi_z\in \mathbb{C}^{A_0}$ is  $\phi_z=\lim_{n\to\infty}z^n\psi_n$ and called the stationary state. 
If we insert the constant inflow (which is equivalent to $z=1$) into the internal graph, then the stationary state is described by a potential function which satisfies the Poisson equation with respect to the Laplacian matrix~\cite{HSS1}.  On the other hand, if we insert the inflow whose signature is alternatively changed with respect to the time step; that is, $z=-1$,   then the stationary state has a potential function which satisfies the Poisson equation with respect to the {\it signless} Laplacian matrix~\cite{HSS2}. 
By setting the inflow $z=e^{i\xi}$, the former corresponds to $\xi=0$ while the latter corresponds to $\xi=\pi$. 
Now a natural problem is to connect between them by considering the general frequency $\xi\in[0,\pi]$. In this paper, 
we consider such a type of dynamics, which is precisely given in (\ref{eq:dynamics}) on the tailed graph $\tilde{G}$ with the internal graph $G_{0}$ in Section~\ref{subsec:setting} under the general inflow condition in (\ref{eq:initialstate}). To characterize the behavior of this dynamics,
we introduce the generalized Laplacian operator parameterized by $z$; if $z=\pm 1$, then the Laplacian and the signless Laplacian matrices are reproduced, respectively. More precisely, the generalized Laplacian matrix is defined as follows. 
\begin{definition}[Generalized Laplacian matrix]
Let $G_0=(V_0,A_0)$ be the internal graph with the boundary $\delta V_0\subset V_0$, which is the symmetric digraph. 
Let $j_{\pm}(z)=(z\pm z^{-1})/2$ for $z\neq 0$. 
The generalized Laplacian matrix 
on $\mathbb{C}^{V_0}$
is defined by
\[ L_{z}=M_0-j_+(z) D_0+j_-(z)\Pi_{\delta V_0} \]
for $z\in \delta \mathbb{D}$. 
Here $M_0$, $D_0$ are the adjacency matrix and the degree matrix of $G_0$, and $\Pi_{\delta V_0}$ is the projection matrix onto $\delta V_0$; that is, $(\Pi_{\delta V_0})_{u,v}=1$ if and only if $u=v\in\delta V_0$. 
\end{definition}
\begin{figure}[hbtp]
    \centering
    \includegraphics[keepaspectratio, width=110mm]{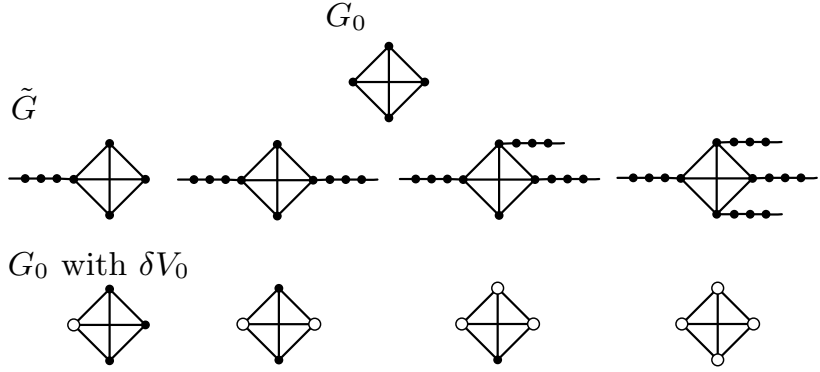}
    \caption{
    Examples of the tailed graph $\tilde{G}$ with an internal graph $\tilde{G}_{0}\simeq K_4$, the complete graph with $4$ vertices, are shown in the second and first row, respectively. In the third row, the $G_{0}$ with $\delta V_{0}$ corresponding to each $\tilde{G}$ is shown: the white vertices belong to $\delta V_{0}$. 
    We depict the cases for the number of boundaries $\ell=1,2,3,4$. The time evolution follows the Grover walk. Note that the time evolution on the tails is free because of a property of the Grover walk. Then once a quantum walker goes out to the tail from the internal graph, it never returns back to the internal, and this situation constitutes the outflow. The initial state for each case $\ell$ is set on the tails so that a quantum walker penetrates the internal graph at every time step, and this situation constitutes the inflow.
    }
    \label{Fig:graphs2}
\end{figure}
We find the circuit equation for the Grover walk; that is, we show that the stationary state with the general $z\in \delta \mathbb{D}$ has the potential function which satisfies the Poisson equation with respect to this generalized Laplacian matrix as follows.
\begin{theorem}[Circuit equation of Grover walk]\label{prop:PoiPoten}
Let $\phi_{e^{i\theta}}\in \mathbb{C}^{A_0}$ be the stationary state with the inflow $\alpha_{in}\in\mathbb{C}^{V_0}$ and its frequency $\theta\in \mathbb{R}$. Then the stationary state  $\phi_{e^{i\theta}}$ has the potential function $\nu_{e^{i\theta}}\in \mathbb{C}^{V_0}$ such that 
\[ i\sin\theta \;\phi_{e^{i\theta}}(a)=e^{i\theta} \nu_{e^{i\theta}}(t(a))-\nu_{e^{i\theta}}(o(a)). \]
Here $\nu_{e^{i\theta}}$ satisfies the following Poisson equation. 
\[ L_{e^{i\theta}} \nu_{e^{i\theta}} =i\sin \theta \bs{\alpha}_{in}, \]
where $L_z$ is the generalized Laplacian matrix on $\mathbb{C}^{V_0}$. 
\end{theorem}
\noindent
If we directly substitute
$\theta=0,\pi$ into the above theorem to recover the results on ~\cite{HSS1,HSS2}, then we notice that $\phi_z(a)$ has an indeterminate form,  
since both $\theta$'s are singular 
points of $L_{e^{i\theta}}$. 
However they turn out to be  
removable, and thus we 
 will continuously connect the stationary states at $\theta=0$ and $\theta=\pi$ with their neighborhoods by using Kato's perturbation theory~\cite{Kato}.

Starting from this theorem, we focus on characterizing the stationary state from the following two quantities. The first is the scattering on the surface. The scattering shows the relation between the inflow and its response to the outside, namely the outflow. In the complete graph, we will see that if we set $\theta=\pi$ or $\theta=\pm \arccos(-1/(N-1))=:\pm\theta_{*}$, then the perfect reflection always occurs if $|\delta V|>1$.  
The second quantity is the comfortability. 
From the classical point of view, the larger the number of drains is, the more difficult it is for the water to pool in the tank in the stationary state because there are so many exits for the water. We set the squared modulus of the stationary state in the internal graph as the comfortability~\cite{HSS2}, which corresponds to the amount of the water in the tank. 
In this paper, we will see that the worst setting of the boundary for the comfortability of quantum walk on the complete graph is that which joins the tail to every vertex; this conclusion accords with the above intuition.
However, we show that if we set some special frequency of the inflow, then the more the boundaries, whose number is strictly less than the number of vertices, are in the graph, the more {\it comfortable} the quantum walker feels. 
This means that if we add the ``unnecessary" tail to the best situation in regard to the comfortability, then the comfortaiblity of the quantum walk suddenly declines to the worst level (see Figure~\ref{Fig:connection}).  

These interesting phenomena are based on  
the following expressions for the scattering and comfortability on a given general connected finite graph $V_{0}$.
Recall that both these quantities yield  
information in the long time limit. 
The former indicates 
the information on the surface 
of the internal graph $V_{0}$ with $\delta V_{0}$; that is, 
how the inflow and outflow are related. In contrast, 
the latter, whose precise definition is given in Definition~\ref{def:2.3} in Section~\ref{subsec:2.6}, indicates the information
on how much the total energy is stored inside $V_{0}$. 
\begin{theorem}[Scattering matrix]\label{cor:scattering}
Let $\bs{\alpha}_{\delta}, \bs{\beta}_\delta\in \mathbb{C}^{\delta V_0}$ represent the inflow and outflow in the long time limit. 
Then we have 
\[ \bs{\beta}_{\delta}=S_{e^{i\theta}}\bs{\alpha}_{\delta}, \]
where $S_{e^{i\theta}}$ is the unitary matrix on $\mathbb{C}^{\delta V_0}$ described by
\[
S_{z}=z\;\chi_{\delta {V_0}}\; \left\{2i\sin\theta L_{z}^{-1}-I_{V_0}\right\}\;\chi_{\delta V_0}^*,
\]
where $\chi_{\delta V_0}$ is isomorphic to the $|\delta V_0|\times |V_0|$ matrix,   $[\;I_{\delta V_0} \;|\; \bs{0}\;]$.  
\end{theorem}
\begin{theorem}[Comfortability]\label{cor:comf}
The comfortability for $z=e^{i\theta}$ is described by \begin{equation}
\mathcal{E}_{z}=
\langle  L_{z}^{-1}\;\bs{\alpha}_{in},(D_0-\cos\theta M_0)L_{z}^{-1}\;\bs{\alpha}_{in} \rangle.
\end{equation}
\end{theorem}
\noindent In the expressions of the scattering matrix and the comfortability, we use the inverse of $L_{e^{i\theta}}$. There exists a frequency $\theta$ such that $\det(L_{e^{i\theta}})=0$ up to the graph structure and the boundary $\delta V_0$. 
Indeed, for the complete graph, there exists such a frequency at $\theta=0,\pi,\pm \arccos(-1/(N-1))$, and we will see a special response of the scattering matrix and comfortability. For example, at $\theta=\pm \theta_*$, the perfect reflection occurs, and also the comfortability is the largest if the number of boundaries and $N$ are sufficiently large. 
However we will show that such a frequency is a removable singularity from the view point of the function of $z=e^{i\theta}$ by extending the domain of $z$ analytically to $\{z\in \mathbb{C} \;:\; |z|<s\}$ for some real value $s>1$. 

The rest of this paper is organized as follows. 
Section~2 is devoted to the setting of our quantum walk model. In particular, we define the stationary state of our quantum walk model and the generalized Laplacian matrix which plays an important role in describing the stationary state and the scattering matrix. 
In Section 3, we first give the proof of Theorem 1.1. 
If $L_z$ is invertible, then the stationary state $\phi_z$
can be immediately expressed by $L_z^{-1}$ and $\bs{\alpha}_{in}$. 
Thus we clarify that the  generalized Laplacian matrix is not invertible if and only if (i) $z=1$ or (ii) $z=-1$ and $G_0$ is a bipartite graph or (iii) $(z+z^{-1})/2$ is the eigenvalue of the underlying random walk whose eigenvector has no overlap with the boundary.  
We also show that such eigenvalues are a   spectrum of the truncated operator with respect to the internal graph.
Analyzing the properties of singularities of $L_z$, 
we show that such singularities are removable for $\phi_z$; 
we can express the stationary state $\phi_z$ as 
$\partial^*L_z^{-1}\bs{\alpha}_{in}$ continuously 
with respect to $z$ ($|z|=1$). 
Then we give the proofs of Theorems 1.2 and 1.3. 
In Section 5, for the complete graph with $N$ vertices, 
we give the concrete forms of the scattering matrix and the comfortability in terms of $N$, the frequency $\theta$ and the number of tails $\ell$. In addition, for every $\ell$, we illustrate the relationship between the transmitting rate and $\theta$, and between the comfortability and $\theta$. These examples provide fruitful information in regard to the effects of quantum walks.
    \label{Fig:graphs}
\section{Setting of our quantum walk model}

\subsection{Graph notation}
The symmetric digraph $G=(V,A)$ is the digraph such that there exists an inverse arc $\bar{a}\in A$ for any $a\in A$.  
The origin and terminal vertices of $a\in A$ are denoted by $o(a)$ and $t(a)$, respectively. 
Remark that $o(\bar{a})=t(a)$ and $|\bar{a}|=|a|$ for any $a\in A$.  
The degree of $u\in V$ is the number of arcs whose terminal vertices are $u$, that is
    \[\deg_G(u):=\#\{ a\in A  \;|\; t(a)=u\}=\#\{a\in A \;|\; o(a)=u\}. \]
\subsection{Tailed graph}\label{subsec:setting}
Let $G_0=(V_0,A_0)$ be a connected and finite symmetric digraph. 
We set arbitrary chosen vertex subset of $V_0$ as the set of the boundary vertices by $\delta V_0\subset V_0$.  We call it the surface of $G_0$.
We connect the semi-infinite length path to each boundary vertex of $\delta V=\{u_1,\dots,u_r\}$. 
The semi-infinite length path connected to the boundary vertex $u_j$ is called the tail connected to $u_j$ and described by $Tail(j)$. 
The resulting semi-infinite graph is denoted by $\tilde{G}=(\tilde{V},\tilde{A})$ and called the tailed graph of $G_0=(V_0,A_0)$ with the boundary set $\delta V_0$. 
The arc set of the tail connected to $u_j$ is denoted by $A_j$ $(j=1,\dots,r)$. 
We set $\tilde{d}(u)$ and $d(u)$ as the degree of $u\in V_0$ in $\tilde{G}$ and $G_0$, respectively. 
\subsection{Random walk}
Let $\Omega$ be a countable set. We define $\mathbb{C}^\Omega$ by the vector space whose standard basis vectors are labeled by $\Omega$. 
The adjacency matrix and the degree matrix of $G$ on $\mathbb{C}^{V}$ are defined by 
\[ (Mf)(u)=\sum_{o(a)=u} f(t(a)),\;\;(Df)(u)=\deg(u)f(u)  \]
for any $f\in \mathbb{C}^V$ and $u\in V$. 
Moreover the probability transition matrices on $\tilde{G}$ and $G_0$ are denoted by 
\[ \tilde{P}=D^{-1}M(\tilde{G}),\;P_0=D_0^{-1}M(G_0), \]
respectively.
Here $M(\tilde{G})$ and $M(G_0)$ are the adjacency matrices of $\tilde{G}$ and $G_0$, respectively, and $D_0$ is the degree matrix of $G_0$. 
We will set $M_0:=M(G_0)$ in short. 
\subsection{Grover walk}
If the scattering rule at each vertex is described by the Grover matrix in the quantum walk dynamics, we call it the Grover walk. 
We will explain this in more detail in the following. 
The Grover matrix of degree $d$ is denoted by \[ \mathrm{Gr}(d)=\frac{2}{d}J_d-I_d, \]
where $J_d$ is the all $1$ matrix and $I_d$ is the identity matrix.
The time evolution operator of the Grover walk $U:\mathbb{C}^{\tilde{A}}\to \mathbb{C}^{\tilde{A}}$ is defined as follows. Let $\tilde{\psi}_n$ be the $n$-th iteration of the Grover walk with some initial state; that is, 
\[ \tilde{\psi}_{n+1}=U\tilde{\psi}_n. \]
Let the set of arcs whose origin vertices are $u\in \tilde{V}$ be denoted by $\{a_1,\dots,a_\kappa\}$, where $\kappa=\deg_G(u)$. 
For any $\psi\in \mathbb{C}^{\tilde{A}}$ and vertex $u\in \tilde{V}$, the time evolution can be expressed as follows. 
\[ \begin{bmatrix} \tilde{\psi}_{n+1}(a_1) \\ \vdots \\ \tilde{\psi}_{n+1}(a_\kappa) \end{bmatrix}=\mathrm{Gr}(\kappa)\begin{bmatrix} \tilde{\psi}_n(\bar{a}_1) \\ \vdots \\ \tilde{\psi}_n(\bar{a}_\kappa) \end{bmatrix}.  \]
Since the Grover matrix is unitary, the unitarity of the total time evolution operator $U$ is ensured; that is, $UU^*=U^*U=I_{\tilde{A}}$. 
Another expression for the time evolution $U$ is
\begin{equation}\label{eq:dynamics} (U\tilde{\psi})(a)=\frac{2}{\tilde{d}(o(a))}\sum_{t(b)=o(a)}\tilde{\psi}(b)-\tilde{\psi}(\bar{a})
\end{equation}
for any $\tilde{\psi}\in\mathbb{C}^{\tilde{A}}$.
This expression will be a starting point for the proof of our theorems. 
\subsection{Initial state and stable orbit}
The initial state is set so that the support is included in all the arcs of the tails whose directions head to the original graph $G_0$. More precisely, setting the parameters of the input by $\alpha_j\in \mathbb{C}$ ($j=1,\dots, r$), we describe the initial state with the complex number $z\in \delta \mathbb{D}:=\{w\in \mathbb{C} \;|\; |w|=1\}$ by 
\begin{equation}\label{eq:initialstate} \tilde{\psi}_0(a)=\begin{cases} \alpha_j z^{-k} & \text{: $a\in A_j$ and   $\dist(o(a),u_j)>\dist(t(a),u_j)=k$,  }\\
&\text{\qquad\qquad\qquad\qquad\qquad\qquad\qquad $(j=1,\dots,r$; $k=0,1\dots)$,}\\
0 & \text{: otherwise.}\end{cases}
\end{equation}
Although our interest is the initial state with the parameter $z\in \delta\mathbb{D}$, we will often  consider the initial state with the parameter  $z\in\mathbb{C}\setminus\{0\}$ in (\ref{eq:initialstate}) for our discussion. 
Note that the quantum walker on each tail is {\it free} since 
\[\mathrm{Gr}(2)=\begin{bmatrix}0& 1\\1 & 0\end{bmatrix}. \]
Then at every time step $n$, the internal graph $G_0$ receives inflow from each tail $\alpha_jz^{-n}$. 
On the other hand, once a quantum walker in the internal graph goes out to a tail, then it never goes back. Such a quantum walker is regarded as outflow to the tail. The convergence to a stationary state of the quantum walk has been ensured in \cite{HS}.   
Our object in this paper is to know the 
long time behavior of the dynamics 
$$
\tilde{\psi}_{n}
=U^{n}\tilde{\psi}_0
$$
with (\ref{eq:dynamics}) and (\ref{eq:initialstate}).
\begin{theorem}[\cite{HS}]\label{thm:existance}
For $z\in \delta \mathbb{D}$, 
let $\tilde{\phi}_n:=z^n\tilde{\psi}_n$. Then we have 
\begin{align*}
    \exists\lim_{n\to\infty}\tilde{\phi}_n &:=\tilde{\phi}_z, \\
    U\tilde{\phi}_z &= z^{-1}\tilde{\phi}_z.
\end{align*}
\end{theorem}
Our interest is the properties of the stationary state $\tilde{\phi}_z$ for the following reason. 
Since the dynamics on the tails are trivial, we will focus on the dynamics on the internal graph. To this end, let us set 
$\chi_{A_0}: \mathbb{C}^{\tilde{A}}\to\mathbb{C}^{A_0}$ as the restriction to the internal graph  such that 
\[ (\chi_{A_0}\tilde{\psi})(a)=\tilde{\psi}(a). \]
Then the adjoint is described by
\[ (\chi_{A_0}^*\psi)(a) = \begin{cases} \psi(a) & \text{: $a\in A_0$,}\\ 0 & \text{: otherwise.} \end{cases} \]
Let the restriction to $\mathbb{C}^{A_0}$ of $\psi_n$ be denoted by $\psi_n:=\chi \tilde{\psi}_n$. Then we have 
\begin{equation}\label{eq:iteration_psi_n}
    \psi_{n+1}=E_{PON}\psi_n+z^{-n}\rho,\;\psi_0=0, 
\end{equation} 
where $E_{PON}=\chi_{A_0}U\chi_{A_0}^*$, and $\rho=\chi_{A_0} U\psi_0$. Details are found in \cite{HS}. 
It is easy to see that $\mathrm{spec}(E_{PON})\subset \mathbb{D}:=\{z\in \mathbb{C} \;|\; |z|\leq 1\}$. 
Let us extend the domain of $z$ to $|z|<1/|\lambda_1|$ in the following. Here $\lambda_1\in \mathrm{spec}(E_{PON})$ has the largest absolute value  in $\mathrm{spec}(E_{PON})\setminus \delta \mathbb{D}$.  
From the time evolution of the original time sequence $\psi_n$ in  (\ref{eq:iteration_psi_n}), $\psi_n$ obviously does not converge to a fixed point if $z\in \delta \mathbb{D} \setminus \{1\}$ because of the oscillation of the outflow represented by  $z^{-n}\rho$.  
Thus instead of $\psi_n$, we set $\phi_n:=z^n\psi_n$ to cancel the oscillation. Indeed, (\ref{eq:iteration_psi_n}) is written by 
\begin{equation}\label{eq:iteration_phi_n}
    z^{-1}\phi_{n+1}=E_{PON}\phi_n+\rho,\;\phi_0=0.
\end{equation} 
By \cite{HS} and Proposition~\ref{prop:analytic} in this paper, $\phi_n$ converges to a fixed point when $|z|<1/|\lambda_1|$, that is, 
\[ \exists ! \lim_{n\to\infty}\phi_n=:\phi_z\]
and $\phi_z$ is a solution of the following equation: 
\begin{equation}\label{eq:limitz}
    z^{-1}\phi_z=E_{PON}\phi_z+\rho. 
\end{equation} 
This is equivalent to stating that for any $\epsilon>0$, there exists $n_0$ such that 
\begin{align}\label{eq:limitz0}
 |z|^n \; || \psi_n-z^{-n}\phi_z ||<\epsilon
\end{align}
for any $n>n_0$. If $1\leq |z|< 1/|\lambda_1|$, this includes 
\begin{equation}\label{eq:stableorbit}
     || \psi_n-z^{-n}\phi_z ||<\epsilon.  
\end{equation}
Then by (\ref{eq:stableorbit}), the sequence $\{\psi_n\}$ is attracted to the stable orbit  $\mathcal{A}_z:=\{z^{-n}\phi_z \;|\; n\in \mathbb{N}\}$; indeed inserting $z^{-n}\phi_z=:\psi^{(n)}\in \mathcal{A}_z$ into LHS of (\ref{eq:iteration_psi_n}) as $\psi_n$, we have $E\psi^{(n)}+z^{-n}\rho=\psi^{(n+1)}\in \mathcal{A}_z$ for $1\leq |z|<1/|\lambda_1|$ by (\ref{eq:limitz}).
Thus once we obtain $\phi_z$, the limit behavior of the original $\psi_n$ can be described in the stable orbit $\mathcal{A}_z$. 
In particular, if $|z|=1$, then $|\phi_z(a)|^2=\lim_{n\to^\infty}|\psi_n(a)|^2$ for any $a\in A_0$.  On the other hand, if $1<|z|<1/|\lambda_1|$, then $\psi_n$ degenerates to $0$ as $n\to\infty$. 
This is the reason for focusing on $\phi_z$ or $\tilde{\phi}_z$ with $z\in \delta \mathbb{D}$. 
Note that for a technical reason, we sometimes extend the domain of $z$ to $|z|<1/|\lambda_1|$ in this paper. 

Let us return to the original dynamics on 
$\tilde{G}=(\tilde{V},\tilde{A})$ with the initial condition (\ref{eq:initialstate}). From the observation stated above, 
we can easily construct $\tilde{\psi}^{(n)}$ on $\tilde{G}$, 
by extending $\psi^{(n)}\in \mathcal{A}_z$ on $G_{0}$, such that  
$\tilde{\psi}^{(n)}$ satisfies 
$$
\tilde{\psi}^{(n+1)} = 
z^{-1} \tilde{\psi}^{(n)} = 
U\tilde{\psi}^{(n)}
$$
and 
$$
\tilde{\psi}^{(n)}(a)=\alpha_j z^{-(k+n)} 
$$
for every $n=0,1,2,\dots$ and 
 every "incoming arc" $a$ on tails;
 that is, 
 $a\in A_j$ such that    $\dist(o(a),u_j)>\dist(t(a),u_j)=k$ 
$(j=1,\dots,r$; $k=0,1\dots)$.
For $|z|=1$, the stationary state for $U$ 
with (\ref{eq:initialstate}) does not exist in general, but we can see that the long time behavior of the dynamics with (\ref{eq:initialstate}) tends to the 
stable orbit 
$\tilde{\mathcal{A}}_z=\{\tilde{\psi}^{(n)}\}$. Naturally, for each 
arc $a\in\tilde{A}$, the value  $|\tilde{\psi}^{(n)}(a)|^{2}$ does not depend on $n$: in particular, 
$|\tilde{\psi}^{(n)}(a)|^{2} = 
|\phi_z(a)|^{2}$ for every arc $a\in A_{0}$ in the internal graph $G_{0}$. 
In this sense, 
we can catch the ``stationary state" for 
this kind of dynamics on $\tilde{G}$; 
thereafter we mainly  
discuss the characterization of 
$\phi_z$ on $G_{0}$ with 
the boundary $\delta V_0$, 
whose precise definition is seen in the next subsection.
\subsection{Definitions of the stationary state, the scattering matrix and generalized Laplacian matrix}\label{subsec:2.6}
\begin{definition}[Stationary state]
Let $\tilde{\psi}_n$ be the $n$-th iteration of Grover walk $U$ on $\tilde{G}$ with the initial state $\tilde{\psi}_0$ (\ref{eq:initialstate}).  
The stationary state of this quantum walk is defined by 
\[ \tilde{\phi}_z=\lim_{n\to\infty}z^n\tilde{\psi}_n. \]
In particular, the restriction of $\tilde{\phi}_z$ to $V_0$ is denoted by $\phi_z\in \mathbb{C}^{A_0}$; that is,
\[\phi_z=\chi_{A_0}\tilde{\phi}_z. \]
\end{definition}
Let the set of boundary vertices be $\delta V_0:=\{u_1,\dots,u_r\}\subset V_0$ and the arc of the tail whose terminal vertex is $u_j$ by $e_j$ $(j=1,\dots,r)$.  
Let the input from the tail $\bs{\alpha}_{in}\in\mathbb{C}^{V_0}$ be defined by \begin{equation}\label{eq:input}
\bs{\alpha}_{in}(u)=\begin{cases} \tilde{\phi}_z(e_j)=\tilde{\phi}_0(e_j) & \text{: $t(e_j)=u$ ($j=1,\dots,r$)}\\
0 & \text{: otherwise}\end{cases}
\end{equation}
while the output to the tail
$\bs{\beta}_{out}\in\mathbb{C}^{V_0}$ be defined by
\[\bs{\beta}_{out}(u)=\begin{cases} \tilde{\phi}_z(\bar{e}_j) & \text{: $t(e_j)=u$ ($j=1,\dots,r$)}\\
0 & \text{: otherwise.}\end{cases}
\]
Let us set 
 $\chi_{\delta V_0}: \mathbb{C}^{V_0}\to \mathbb{C}^{\delta V_0}$ by 
\[ (\chi_{\delta V_0}f)(u_j)=f(u_j)\;(j=1,\dots,r). \]
The adjoint $\chi_{\delta V_0}^*:\mathbb{C}^{\delta V_0}\to \mathbb{C}^{V_0}$ is described by 
\[ (\chi_{\delta V_0}^*f_o)(u)=\begin{cases} f_o(u) & \text{: $u\in \delta V_0$} \\ 0 & \text{: $u\notin \delta V_0$} \end{cases}
\]
We put $\bs{\alpha}_\delta=\chi_{\delta V_0}\bs{\alpha}_{in}$ and $\bs{\beta}_\delta=\chi_{\delta V_0}\bs{\beta}_{out}$. 
In~\cite{FelHil1,HS}, the existence of the following matrix is ensured and this matrix is a unitary matrix on $\mathbb{C}^{\delta V_0}$.
\begin{definition}[Scattering matrix]
The scattering matrix on the surface is defined by \[ \bs{\beta}_\delta=S(z;G_0;\delta V_0)\bs{\alpha}_\delta, \]
which is independent of the in- and outputs $\bs{\alpha}_\delta$ and $\bs{\beta}_{\delta}$. 
\end{definition}
We are also interested in the stationary state in the internal graph. So we define the following which may be interpreted as how much quantum walker feels comfortable in the internal graph. 
\begin{definition}[Comfortability]\label{def:2.3}
For $z\in \delta\mathbb{D}$,
the comfortability is defined by 
\[ \mathcal{E}_z=\frac{1}{2}||\phi_z||^2_{\ell^2}=\frac{1}{2}\sum_{a\in A_0}|\phi_z(a)|^2. \]
\end{definition}

We will consider how the information on the internal graph appears in the stationary state, the scattering matrix and the comfortability. 
To this end, let us 
introduce a key matrix on $\mathbb{C}^{V_0}$. 
\begin{definition}[Generalized Laplacian matrix]\label{def:jz}
For $z\in \mathbb{C}\setminus \{0\}$,  
\[ j_{\pm}(z):= \frac{z\pm z^{-1}}{2}.  \]
Let $L_z:=L(z;G_0;\delta V_0)$ be defined by
    \[ L_z:= M_0-j_+(z)D_0+j_-(z)\Pi_{\delta V_0}. \]
\end{definition}
The matrix  $L_z$ with parameter $z$ on $\mathbb{C}^{V_0}$ reproduces the Laplacian matrix ($z=1$) and the signless Laplacian matrix ($z=-1$) of $G_0$ by the parameter $z\in \delta\mathbb{D}$. Then we call $L_z$ the generalized Laplacian matrix. The generalized Laplacian matrix has the information on the graph and its boundaries and plays an important role in expressing the stationary state of the Grover walk. 
\section{Main theorems}
Recall that the stationary state of the quantum walk is denoted by $\tilde{\phi}_z$ with the initial state expressed by $\bs{\alpha}_{in}\in \mathbb{C}^{V_0}$.  
The stationary state restricted to $G_0$ is denoted by
\[ \phi_z:=\chi_{A_0}\tilde{\phi}_z. \]

Let $\tilde{\nu}_z\in \mathbb{C}^{\tilde{V}}$ denote  the average of the stationary state $\tilde{\phi}_z(a)$ for $t(a)=u$, that is,  
\begin{equation}\label{eq:nuz}
\tilde{\nu}_z(u)= \frac{1}{\deg_{\tilde{G}}(u)}\sum_{t(a)=u\text{ in $\tilde{G}$}} \tilde{\phi}_z(a)  \end{equation}
for every $u\in \tilde{V}$. 
Let $\chi_{V_0}: \mathbb{C}^{\tilde{V}}\to\mathbb{C}^{V_0}$ be the restriction to the internal graph  such that 
\[ (\chi_{V_0}f)(u)=f(u) \]
and the adjoint is described by
\[ (\chi_{V_0}^*g)(u) = \begin{cases} g(u) & \text{: $u\in V_0$,}\\ 0 & \text{: otherwise.} \end{cases} \]
The restriction of $\tilde{\nu}_z$ to $G_0$ is denoted by 
\[ \nu_z := \chi_{V_0}\tilde{\nu}_z. \]

The following theorem is the starting point of all the considerations for the stationary state:
the function $\nu_z$ is a kind of potential function of the stationary state $\phi_z$ and satisfies a kind of Poisson equation. 
\begin{reptheorem}{prop:PoiPoten}[Circuit equation for the Grover walk]
Let $z\in \delta \mathbb{D}$. 
Assume $G_0=(V_0,A_0)$ is a connected and symmetric digraph and $\tilde{G}=(\tilde{V},\tilde{A})$ is the tailed graph of $G_0$ with the boundary vertex set $\delta V_0$. 
Then the function $\tilde{\nu}_z\in \mathbb{C}^{\tilde{V}}$ is the twisted potential function of $\tilde{\phi}_z\in \mathbb{C}^{\tilde{A}}$ such that
\begin{equation}\label{eq:startingpoint2}
    j_-(z)\;\tilde{\phi}_z(a)=z\;\tilde{\nu}_z(t(a))-\tilde{\nu}_z(o(a)) 
\end{equation}
for any $a\in \tilde{A}$.
Here $\nu_z\in \mathbb{C}^{V_0}$ satisfies the Poisson equation 
\begin{equation}\label{eq:startingpoint}
L_z\; \nu_z = j_-(z)\;\bs{\alpha}_{in}
\end{equation}
for every $u\in V_0$. 
\end{reptheorem}
\begin{proof}
By the definition of the Grover walk, the generalized  eigenequation in Theorem~\ref{thm:existance} implies  
\begin{equation}\label{eq:eigeneq}
    \frac{1}{2}\left(\; \tilde{\phi}_z(a)+z\tilde{\phi}_z(\bar{a})  \;\right)= z \tilde{\nu}_z(o(a))
\end{equation}
for every $a\in \tilde{A}$. Since (\ref{eq:eigeneq}) also holds at the inverse of the arc $a\in \tilde{A}$, we have 
\begin{equation}\label{eq:eigeneq2}
    \frac{1}{2}\left(\; \tilde{\phi}_z(\bar{a})+z\tilde{\phi}_z(a)  \;\right)= z \tilde{\nu}_z(t(a)).
\end{equation}
Equations (\ref{eq:eigeneq}) and (\ref{eq:eigeneq2}) imply 
\begin{equation}\label{eq:eigeneq3}
    j_-(z)\;\tilde{\phi}_z(a) = z \tilde{\nu}_z(t(a))-\tilde{\nu}_z(o(a)). 
\end{equation}
Then the proof of (\ref{eq:startingpoint2}) is completed. 
Taking the summation of (\ref{eq:eigeneq3}) over $t(a)=u$, we obtain the following generalized eigenequation of the random walk induced by the generalized eigenequation of the quantum walk
\begin{equation}\label{eq:Randomwalk} j_+(z)\tilde{\nu}_z(u)=(\tilde{P}\tilde{\nu}_z)(u)
\end{equation}
for every $u\in \tilde{V}$. 
Recall that the set of boundary vertices are denoted by $\delta V_0:=\{u_1,\dots,u_r\}\subset V_0$ and the arc of the tail whose terminal vertex is $u_j$ by $e_j$ $(j=1,\dots,r)$. Note that $\phi_z(e_j)=\alpha_j$ ($j=1,\dots,r$), which is the input to the internal graph. 
Let us consider (\ref{eq:Randomwalk}) on the boundary as follows. 
By (\ref{eq:eigeneq3}), we have 
\begin{equation}\label{eq:boundary}
    \tilde{\nu}_z(o(e_j))=z\; \tilde{\nu}_z(u_j)-j_-(z)\;\alpha_j,\text{ ($j=1,\dots,r$).}
\end{equation}
Putting $P'=\chi_{V_0}\tilde{P}\chi_{V_0}$, which is the principal submatrix of $\tilde{P}$ with respect to the internal graph, we have an equivalent expression of (\ref{eq:Randomwalk}) as follows: 
\begin{align}\label{eq:internal}
    j_+(z)\nu_z(u_j)&=\frac{1}{\tilde{d}(u_j)}\tilde{\nu}_z(o(e_j))+(P'\nu_z)(u_j) \text{ for any $u_j\in \delta V_0$,}\\
    j_+(z)\nu_z(u)&=(P'\nu_z)(u) \text{ for any $u\in V_0\setminus \delta V_0$.}
\end{align}
Inserting (\ref{eq:boundary}) into the expression of $\tilde{\nu}_z(o(e_j))$ in (\ref{eq:internal}) and remarking that  $\tilde{\nu}_z(u_j)=\nu_z(u_j)$, we obtain
\begin{equation}\label{eq:RWQW}
    (P'-j_+(z)+z {\tilde{D}}^{-1}\Pi_{\delta V_0})\nu_z=j_-(z){\tilde{D}}^{-1}\bs{\alpha}_{in}. 
\end{equation}
Here $\tilde{D}$ is the degree matrix in $\tilde{G}$, that is,  $(\tilde{D}f)(u)=\tilde{d}(u)f(u)$ for any $u\in V_0$. Since $\tilde{d}(u_j)=d(u_j)+1$ 
for every boundary vertex, note that $\tilde{D}=D_0+\Pi_{\delta V_0}$.
Then (\ref{eq:RWQW}) is equivalent to 
\begin{equation*}
L(z;G_0;\delta V_0)\nu_z = j_-(z)\bs{\alpha}_{in}. 
\end{equation*}
\end{proof}

The above statements (\ref{eq:startingpoint2}) and (\ref{eq:startingpoint}) in Theorem~\ref{prop:PoiPoten} hold for any connected and symmetric digraph $G_0$. The existence of ``$1/j_-(z)$"  depends on the input parameter $z$; that is, $z\notin \{\pm 1\}$,  while that of the inverse of the matrix $L_z$ depends on the internal graph geometry $G_0$ and the choice of the boundary $\delta V_0$. 
Let us set such points by 
\[\mathbb{B}_*:=\{ z\in \mathbb{C}\setminus\{0\} \;:\; \det(L_z)=0\}\cup \{\pm 1\}.   \]
The following theorem characterizes $\mathbb{B}_*$ by the spectrum of the principal submatrix of $U$, $E_{PON}:=\chi U\chi^*$~\cite{HS}. 
\begin{theorem}\label{thm:removable}
Let $E_{PON}$ be the principal submatrix of $U$ with respect to $A_0$; that is, $E_{PON}=\chi U\chi^*$. 
Let us denote the set $\mathrm{spec}^{\star}(E_{PON})$ as 
\[ \mathrm{spec}^{\star}(E_{PON}):=\{z^{-1} \;|\; z\in \mathrm{spec}(E_{PON})\setminus\{0\}\}. \]
Then we have 
\begin{equation}\label{eq:B*} \mathrm{spec}^{\star}(E_{PON})\cup\{\pm 1\}=\mathbb{B}_*. 
\end{equation}
In particular, 
\begin{equation}\label{eq:B*1}
\mathbb{B}_*\cap \delta\mathbb{D}
= j_+^{-1}(\sigma_{per})\cup \{\pm 1\}. 
\end{equation}
Here 
\[ \sigma_{per}:=\{ \lambda\in \mathrm{spec}(P_0) \;|\; \{f \;:\; \supp(f)\subset V_0\setminus \delta V \} \cap \ker (\lambda-P_0) \neq \emptyset\}. \]
\end{theorem}
\begin{proof}
By \cite{HS}, the eigenvalues of $E_{PON}$ except $\pm 1$ are the zeros of the following equation with respect to $z$: 
\begin{equation}\label{eq:HS} \det((z^2+1)I-2zT-2\Pi_{\delta V_0}D^{-1} )=0. 
\end{equation}
Here $T=D^{-1/2}M_0D^{-1/2}$. 
The LHS can be deformed by 
\begin{align}
\det((z^2+1)I-2zT-2\Pi_{\delta V_0}D^{-1})
&= (2z)^{|V_0|}\det( j_+(z)-T-\Pi_{\delta V_0}(zD)^{-1} ) \notag \\
&= (2z)^{|V_0|}\det(D)^{-1}\det( j_+(z)D-M_0-z^{-1}\Pi_{\delta V_0} ) \notag \\
&=  (-2z)^{|V_0|}\det(D)^{-1}\det (j_+(z)D_0-M_0+j_-(z)\Pi_{\delta V_0}) \notag \\
&= (-2z)^{|V_0|}\det(D)^{-1}\det(L_{z^{-1}}). 
\label{eq:gen}
\end{align}
Here in the second and third equalities, we used $T=D^{-1/2}M_0D^{1/2}$ and  
$D=D_0+\Pi_{\delta V_0}$, respectively; the final equality derives from $j_-(z^{-1})=-j_-(z)$.  
Then we have finished the proof of (\ref{eq:B*}). 
For (\ref{eq:B*1}), since $E_{PON}$ is a real matrix, we have
\begin{align*}
    \mathrm{spec}^\star(E_{PON})\cap \delta \mathbb{D} &= \mathrm{spec} (E_{PON}) \cap \delta \mathbb{D}.
\end{align*}
By \cite{HS}, 
\[(\mathrm{spec} (E_{PON}) \cap \delta \mathbb{D}) \setminus \{\pm 1\}
= j_+^{-1}(\sigma_{per}), \]
which gives the desired conclusion.
\end{proof}

\begin{remark}
By (\ref{eq:HS}), if $0\in \mathrm{spec}(E_{PON})$, then $\det(I-2\Pi_{\delta V_0}D^{-1})$ must be $0$. Since $I-2\Pi_{\delta V_0}D^{-1}$ is a diagonal matrix, there must exist $0$ in the diagonal entries corresponding to $\delta V$ vertices. 
Then the definition of $D$ implies that 
the internal graph $G_0$ has at least one leaf in $\delta V$. 
Conversely, it is easy to see that if the internal graph $G_0$ has a leaf in $\delta V$, then $0\in \mathrm{spec}(E_{PON})$. 
Thus $0\in \mathrm{spec}(E_{PON})$ if and only if the internal graph $G_0$ has a leaf in $\delta V$. 
\end{remark}
\begin{corollary}
If $G_0$ is $\kappa$-regular and $\delta V_0=V_0$, then 
\[ \sigma(E_{PON})\setminus \delta\mathbb{D} \subset \left\{ w\in \mathbb{C} \;\bigg|\; |w|=\sqrt{\frac{\kappa-1}{\kappa+1}} \right\} \cup \left\{ w\in\mathbb{R} \;\bigg|\; \frac{\kappa-1}{\kappa+1}\leq |w|<1 \right\}, \]
and $E_{PON}$ has an eigenvalue in the second term if and only if  
\[\sigma (P_0) \cap \{(-1,-1+1/\kappa)\cup(1-1/\kappa, 1) \}\neq \emptyset. \]
\end{corollary}
\begin{proof}
From (\ref{eq:gen}), we have 
\begin{align}
    z\in \mathrm{spec}(E_{PON})\setminus\{\pm 1\}
    & \Leftrightarrow 
    \det( (\kappa j_+(z)+j_-(z))I -M_0)=0.
\end{align}
This implies that an eigenvalue $\mu\in \mathrm{spec}(M)$ is mapped to 
\begin{align*}
    z &=\frac{\mu\pm \sqrt{\mu^2-(\kappa^2-1)}}{\kappa+1} \\
      &= \sqrt{\frac{\kappa-1}{\kappa+1}}\left\{ \frac{\mu}{\sqrt{\kappa^2-1}}\pm \sqrt{\left(\frac{\mu}{\sqrt{\kappa^2-1}}\right)^2-1} \right\}
\end{align*} 
as eigenvalues of $E_{PON}$. 
If $\mu^2\geq \kappa^2-1$, then $z\in \mathbb{R}$. Moreover, since $\sqrt{\kappa^2-1}\leq |\mu|\leq \kappa$, it is easy to see that $(\kappa-1)/(\kappa+1)\leq |z|<1$. 
On the other hand, if $\mu^2\leq \kappa^2-1$, then $\mu$ is on the circle with the radius $\sqrt{(\kappa-1)/(\kappa+1)}$. 
\end{proof}
\begin{remark}
Assume $G_0$ is $\kappa$-regular and $\delta V_0=V_0$. 
If the internal graph $G_0$ is the Ramanujan graph with $\kappa\geq 3$; that is, $|\lambda|\leq 2\sqrt{\kappa-1}/\kappa$ for any eigenvalue  $\lambda$ of $P_0$ except $1$, then all the eigenvalues of $E_{PON}$ in $\mathbb{D}\setminus \delta \mathbb{D}$ live on the circle with the radius $\sqrt{(\kappa-1)/(\kappa+1)}$.
\end{remark}

There exist many kinds of graphs such that $j_+^{-1}(\sigma_{per})\neq \emptyset$; for example, 
the complete graph with $|\delta V_0|<|V_0|-1$ (see Section~\ref{sect:example}). 
Thus for such a graph which has the spectrum $z\in j_+^{-1}(\sigma_{per})$, we need to pay attention when expressing the stationary state with the inflow $z\in j_+^{-1}(\sigma_{per})$, since $L_z$ itself becomes non-invertible. 
However, the following proposition is simple but has an important message from \cite{HS}. 
\begin{proposition}\label{prop:analytic}
Let $\lambda_1,\lambda_2,\dots$ be all the eigenvalues of $E_{PON}$ in $\mathbb{D}\setminus \delta \mathbb{D}$ with $1>|\lambda_1|\geq |\lambda_2|\geq \cdots$. 
Then the stationary state $\phi_z$ is analytic with respect to $z$ in $0< |z|<1/|\lambda_1|$. 
\end{proposition}
\begin{proof}
Let us extend the domain of the parameter $z\in \delta \mathbb{D}$ to $z\in \mathbb{C}\setminus\{0\}$ formally. 
The $n$-th iteration of the quantum walk restricted to $G_0$,
$\psi_n$, has the following recursion:  
\[ \psi_{n+1}=E_{PON}\psi_n+z^{-n}\rho,\;\psi_0=0, \]
where $\rho(a)=2\bs{\alpha}_{in}(t(a))/\deg_{\tilde{G}}(t(a))$. 
Then $\phi_n=z^{n}\psi_n$ is expressed as follows: for $n\geq 1$, 
\begin{equation}\label{eq:psin} \phi_n=z(1+zE_{PON}+\cdots+(zE_{PON})^{n-1})\rho. \end{equation}
Let $J(\lambda,k)$ be the Jordan block with the size $k$ and its every diagonal element is $\lambda$ such that \[ J(\lambda,k)=\begin{bmatrix} 
\lambda & 1 & 0 &\cdots & 0 \\
        & \lambda & 1 &  \cdots & 0 \\
        &         &\ddots & \ddots & \\
        &         &       & \lambda & 1 \\
        &         &       &         & \lambda 
\end{bmatrix}. \]
Since the submatrix $E_{PON}$ loses the semi-simpleness in general, it should be decomposed into the Jordan blocks by 
\[ E_{PON} \cong \bigoplus_{\lambda\in \mathrm{spec}(E_{PON})}J(\lambda,k_\lambda). \]
By \cite{HS}, there exits an invertible matrix $R$ such that  
\[ (E_{PON})^{j}\rho=R \left(\bigoplus_{\lambda\in \mathrm{spec}(E_{PON}):|\lambda|<1} J^j(\lambda,k_\lambda)\right)R^{-1}\rho   \]
for any $j$. 
Note that as is expressed in the above equation, the $n$-th iteration $\phi_n$ belongs to the stable eigenspace associated with the eigenvalues $|\lambda|<1$ for any time step $n$~\cite{HS}. 
We should note that 
\[\bs{1}_{\{i\leq j\}}\binom{n}{j}\lambda^{n-j}=(J(\lambda,k)^n)_{i,j}\leq (J(\lambda,k)^n)_{1,k}=\binom{n}{k}\lambda^{n-k}\]
for sufficiently large $n$. 
Then the radius of convergence of the power series of $z$ in (\ref{eq:psin}) is 
\[ \lim_{n\to\infty} \left|\frac{\binom{n}{k}\lambda_1^{n-k}}{\binom{n+1}{k}\lambda_1^{n+1-k}}\right|=1/|\lambda_1|. \]
\end{proof}
\begin{remark}\label{rem:phindomain}
The original $n$-th iteration of the quantum walk, $\psi_n$, can not be defined if the inflow parameter is $z=0$. 
Then  $\phi_n:=z^n\psi_n$ also cannot be defined. 
However once $\phi_n$ is considered as the function of $z$ defined by (\ref{eq:psin}), the domain of every entry of $\phi_z$ can be extended to $0\leq |z|<1/|\lambda_1|$. 
\end{remark}
Using this proposition, we can state that such points in $\mathbb{B}_*$ are the removable singularities. 
Before stating the theorem, let us prepare a few notations. Let $\partial_z: \mathbb{C}^{A_0}\to \mathbb{C}^{V_0}$ be a boundary operator such that 
\[ (\partial_z\psi)(u)=z^*\sum_{t(a)=u} \psi(a)-\sum_{o(a)=u}\psi(a).  \]
Note that its adjoint operator $\partial_z^*: \mathbb{C}^{V_0}\to \mathbb{C}^{A_0}$ is described by
\[ (\partial_z^*f)(a)=zf(t(a))-f(o(a)).  \]
\begin{theorem}[Stationary state]\label{thm:stationarystate}
The function $\varphi:\mathbb{C}\to \mathbb{C}^{A_0}$ 
\begin{equation}\label{eq:stationaryextend} \varphi(z)=\partial_z^*L_z^{-1}\bs{\alpha}_{in} 
\end{equation}
is analytically extendable over $ (\mathbb{C}\setminus \mathbb{B}_*)\cup\{z\in \mathbb{C} \;:\;|z|<1/|\lambda_1|\}$. In particular, 
$\varphi(e^{i\theta})$ is the stationary state of this quantum walk with the frequency of the inflow $\theta$; that is, $\varphi(e^{i\theta})=\phi_{e^{i\theta}}$. 
\end{theorem}
\begin{proof}
It is obvious that the function $\varphi(z)$ is  analytic for any $z\in \mathbb{C}\setminus\mathbb{B}_*=:\mathbb{D}_1$. 
On the other hand, by Proposition~\ref{prop:analytic} and Remark~\ref{rem:phindomain}, $\phi_z$ is analytic in $\{z\in \mathbb{C}\;;\;|z|<1/|\lambda_1|\}=:\mathbb{D}_2$.
Theorem~\ref{prop:PoiPoten} leads to $\phi_z=\varphi(z)$ for every $z\in \delta \mathbb{D}\setminus\mathbb{B}_*$. 
Then the identity theorem implies 
$\varphi(z)=\phi_z$ in $\mathbb{D}_2\setminus \mathbb{B}_*$. Then since $(\delta \mathbb{D}\cap \mathbb{B}_*)\subset \mathbb{D}_2$, all of the siguralities of $\varphi(z)$ on $|z|=1$ are removable. 
Note that $|z|=1$ or $|z|\geq 1/|\lambda_1|$ for every $z\in \mathbb{B}_*$ by Theorem~\ref{thm:removable}. 
By taking the direct continuation $\varphi(z)$ at the domain $\mathbb{D}_2\setminus \mathbb{B}_*$ for $\phi_z$ in the analytic continuation, we can extend the domain of $\phi_z$ to $\mathbb{D}_1\cup  \mathbb{D}_2$.  
\end{proof}
\begin{reptheorem}{cor:scattering}[Scattering matrix]
Let $z\in\delta \mathbb{D}$. 
The scattering matrix is expressed by
\[ S_z=z\;\chi_{\delta {V_0}}\; \left\{2j_-(z)L_z^{-1}-I_{V_0}\right\}\;\chi_{\delta V_0}^*.  \]
Here $z\in \delta\mathbb{D}$ with $\det(L_z)=0$ is a removable singularity. 
\end{reptheorem}
\begin{proof}
From (\ref{eq:eigeneq}), inserting the arc of the tail connecting to the boundary vertex $u_j\in\delta V$ into the arc $a$, we have 
\[ \bs{\beta}_\delta(u_j)=2z\nu_z(u_j)-z\bs{\alpha}_\delta(u_j). \]
If $z\notin \mathbb{B}_*$, then Theorem~\ref{prop:PoiPoten} implies $\nu_z(u_j)=j_-(z)(\chi_{\delta V_0}L^{-1}_z\chi_{\delta V_0}^*\bs{\alpha}_{\delta})(u_j)$.  
Then we have 
\[ \bs{\beta}_\delta=z(2j_-(z)\chi_{\delta V_0}L_z^{-1}\chi_{\delta V_0}^*-I_{\delta V_0})\bs{\alpha}_{\delta}=z\chi_{\delta V_0}(2j_-(z)L_z^{-1}-I_{V_0})\chi_{\delta V_0}^*\bs{\alpha}_{\delta}. \]

Next, let us consider the case for $z\in \mathbb{B}_*\cap \delta \mathbb{D}$. 
Let $u,v\in V_\delta$. 
The stationary state with the inflow $\bs{\alpha}_{in}=\delta_{v}$ is denoted by $\phi_z^{(v)}$.  
The $(u, v)$ element of the scattering matrix $(S_z)_{u,v}$ is described by 
\begin{align*}
(S_z)_{u,v} &=2\langle K^*\delta_u,\; \phi_z^{(u)} \rangle-\delta_{u,v} \\
&= 2\langle \delta_u,\; K\phi_z^{(u)} \rangle-\delta_{u,v},
\end{align*}
where $K:\mathbb{C}^{A_0}\to \mathbb{C}^{V_0}$ such that 
$(K\psi)(u)=(1/\deg_{G_0}(u))\sum_{a\in A_0\;:\;t(a)=u}\psi(a)$, 
and its adjoint is $(K^*f)(a)=(1/\deg_{G_0}(u))f(t(a))$. 
Since $\phi_z$ is analytic in $|z|<1/|\lambda_1|$ by Proposition~\ref{prop:analytic}, 
$s_{u,v}(z):=2\langle \delta_u,\; K\phi_z^{(v)} \rangle-\delta_{u,v}$ is also analytic in $|z|<1/|\lambda_1|$. On the other hand, since $s_{u,v}(z)$ is identical with $(z\chi_{\delta V_0}\{2j_-(z)L_z^{-1}-I_{V_0}\}\chi_{\delta V_0}^*)_{u,v}$ for any $z\in \delta \mathbb{D}\setminus\mathbb{B}_{*}$, the identity theorem implies 
\[s_{u,v}(z)=(z\chi_{\delta V_0}\{2j_-(z)L_z^{-1}-I_{V_0}\}\chi_{\delta V_0}^*)_{u,v}  \]
for any $z\in \partial\mathbb{D}$. 
\end{proof}
In the following, let us explain 
the role of the first factor
``$z$" in the expression for $S_z$. 
The inflow at time $n$ to the vertex $u_j\in \delta V$ is expressed by $\bs{\alpha}_\delta^{(n)}(u_j):=z^{-n}\bs{\alpha}_\delta(u_j)$. 
On the other hand, the outflow from the vertex $u_j\in\delta V$ at the next step $n+1$ can be described by
\begin{align*}
    \bs{\beta}^{(n+1)}_\delta(u_j) &= \frac{2}{\tilde{d}(u)}\sum_{t(a)=u} \tilde{\psi}_n(a)-z^{-n}\alpha_j
\end{align*}  
from the definition of this quantum walk. 
Let us replace $\tilde{\psi}_n$ into an element
of the stable orbit; that is, $z^{-n}\tilde{\phi}_z$. Then in a similar fashion to the proof of Theorem~\ref{cor:scattering}, we have 
\begin{align*}
    \bs{\beta}^{(n+1)}_\delta
    &= z^{-n}\;\chi_{\delta V_0 }(2j_-(z)L_z^{-1}-I_{V_0})\chi_{\delta V_0}^*\bs{\alpha}_\delta. \\
    &= z^{-1} S_z \bs{\alpha}_{\delta}^{(n)}. 
\end{align*}
This means that the matrix $z^{-1}S_z$ gives the response to the input of the previous time, while  the scattering matrix $S_z$ gives the ``snap shot" of the scattering on the surface in the long time limit. 
\begin{reptheorem}{cor:comf}[Comfortability]
The comfortability for $z=e^{i\theta}$ is described by \begin{equation}\label{eq:generalcomftheta}
\mathcal{E}_z=
\langle  L_z^{-1}\bs{\alpha}_{in},(D_0-\cos\theta M_0)L_z^{-1}\bs{\alpha}_{in} \rangle.
\end{equation}
Here, $z\in \delta \mathbb{D}$ with $\det{L_z}=0$ is a removable singularity. 
\end{reptheorem}
\begin{proof}
By Theorem~\ref{thm:stationarystate}, we have 
\begin{align}
    \mathcal{E}_z
    &= \frac{1}{2}\langle \phi_z,\phi_z \rangle \notag \\
    &= \frac{1}{2}\langle \partial_z^* L_z^{-1}\bs{\alpha}_{in},\partial_z^*L_z^{-1}\bs{\alpha}_{in}\rangle \notag \\
    &= \frac{1}{2}\langle L_z^{-1}\bs{\alpha}_{in},\partial_z\partial_z^* L_z^{-1}\bs{\alpha}_{in} \rangle \notag \\
    &= \langle  L_z^{-1}\bs{\alpha}_{in},(\frac{1+|z|^2}{2}D_0-\frac{z+z^{*}}{2} M_0)L_z^{-1}\bs{\alpha}_{in} \rangle \notag\\
    &= |z|\langle  L_z^{-1}\bs{\alpha}_{in},(j_+(|z|)D_0-j_+(z/|z|) M_0)L_z^{-1}\bs{\alpha}_{in} \rangle. \label{eq:ComK}
\end{align}
Then we have (\ref{eq:generalcomftheta}) for $z\in\delta\mathbb{D}\setminus\mathbb{B}_*$. 
Since $\phi_z$ can be defined in $\mathbb{D}_1\cup \mathbb{D}_2$ as an analytic function, $z\in \delta \mathbb{D}\cap \mathbb{B}_*$ is a removable singularity. 
\end{proof}

The above statements of theorems include that the stationary state with the frequency of the inflow $z_*\in \mathbb{B}_*\cap \delta \mathbb{D}$ is obtained by taking the limit of $z\to z_*$. 
The following proposition shows a more direct computational approach to obtain the stationary state for $z_*$ without taking the limit which seems to be more practical when $z_*\neq \pm 1$. 
To this end let us prepare the following inner product. 
\[\langle f,g \rangle_{\pi}:=\sum_{u\in V_0}\bar{f}(u)g(u)d(u) \]
for any $f,g\in \mathbb{C}^{V_0}$. 
On the other hand, the standard inner product is denoted by 
\[\langle f,g \rangle_{2}:=\sum_{u\in V_0}\bar{f}(u)g(u). \]
Note that $\langle f,g \rangle_{\pi}=\langle f,D_0g \rangle_{2}$.
We will discuss the case for $z=\pm 1$ in the next section. 
Now we are ready to give the proposition. 
\begin{proposition}
Let $z\in\delta \mathbb{D}\setminus\{\pm 1\}$. 
Then the potential function $\nu_z$ of the stationary state $\phi_z$ satisfies the following:   
\begin{align}
    L_z\nu_z &= j_-(z)\bs{\alpha}_{in} \notag\\
    \langle \mu_k, \nu_z\rangle_\pi &= 0 \;\;(k\leq \dim (\mathcal{K}_z)). \label{eq:linear2}
\end{align}
Here $\{\mu_k\}$ are the basis of $\mathcal{K}_z:=\{g\in \ker(j_+(z)-P_0) \;|\; \supp (g)\subset V_0\setminus \delta V_0\}$. 
\end{proposition}
This proposition implies that if $z\in \mathbb{B}_*\setminus\{\pm 1\}$, then the potential function
$\nu_z$ is the unique vector which is orthogonal to $\mathcal{K}_{z}$ with respect to the inner product $\langle \cdot,\cdot\rangle_{\pi}$ in the solutions of the linear equation.
\begin{proof}
If $z\notin \mathbb{B}_*$, then the statement is directly obtained by (\ref{eq:startingpoint}) in Theorem~\ref{prop:PoiPoten} since $\mathcal{K}_z=\{\bs{0}\}$.
In the following, we consider the case for  $z\in \mathbb{B}_*$; that is, $z\in j_+^{-1}(\sigma_{per})$. The rank of $L_z$ is $|V_0|-\dim(\mathcal{K}_z)$ by (\ref{eq:B*1}). Thus it is enough to show that $\nu_z\in \ker(L_z)^\perp$
under the inner product $\langle \cdot, \cdot \rangle_\pi$. 
Let us put $\mu\in \ker(L_z)$. Note that 
$(M_0-j_+(z) D_0)\mu=0$ by (\ref{eq:B*1}) in Theorem~\ref{thm:removable} since $\supp(\mu)\cap \delta V=\emptyset$. 
According to \cite{HS}, it holds that  
\begin{equation}\label{eq:permanent}
\langle \partial_z^* \mu, \phi_z \rangle_2=0, 
\end{equation}
which is equivalent to 
\begin{align}\label{eq:center}
    \langle \partial_z^* \mu, \partial_z^* \nu_z \rangle_2=0
\end{align}
by (\ref{eq:startingpoint2}) in Theorem~\ref{prop:PoiPoten}. 
Thus (\ref{eq:center}) is equivalent to  
\begin{align*}
    0&= \langle \partial_z^* \mu, \partial_z^* \nu_z \rangle_2 
     = \langle \partial_z\partial_z^* \mu, \nu_z \rangle_2 \\
     &= 2 \langle\; (-\cos \theta M_0+D_0)\mu, \nu_z \;\rangle_2 \\ 
     &= 2\langle\; (-\cos \theta^2+1)D_0\mu, \nu_z \;\rangle_2 \\
     &= 2(1-\cos^2\theta)\langle \mu, \nu_z \rangle_\pi. 
\end{align*}
Since $\cos^2\theta \neq 1$, we obtain the desired conclusion. 
\end{proof}


  
\section{The stationary state for $z\in \{\pm 1\}$}
In this section, we will obtain the stationary state for $z=\pm 1$ by taking the limits of  $\phi_z$ to $z=\pm 1$ from Theorem~\ref{thm:stationarystate}. 

The case where $G_0$ is non-bipartite with the inflow $z=-1$ is simply obtained as follows. 
From Theorem~\ref{prop:PoiPoten}, in the neighborhood of $z=-1$, we have 
\[ \phi_z=\partial^*_z L_z^{-1}\bs{\alpha}_{in}. \]
Note that if $G_0$ is non-bipartite, the signless Laplacian $L_{-1}=Q$ is invertible. Then 
\[ \phi_{-1}=\lim_{z\to -1}\partial^*_z L_z^{-1}\bs{\alpha}_{in}=\partial^{*}_{-1}Q^{-1}\bs{\alpha}_{in}.  \]
The other cases---that is, $z=1$ or ``$z=-1$ and $G_0$ is  bipartite"---are not so simple because $L_{\pm 1}$ are not invertible.   

The stationary states for $z=\pm 1$ have already been characterized in \cite{HSS1,HSS2} under the direct consideration of equations (\ref{eq:eigeneq}) and (\ref{eq:permanent}). To show the previous results, let us prepare a few notations. If $G_0$ is a bipartite graph with the partite set $V=X\sqcup Y$, it will be useful to use the  notation ``$\flat$" defined as follows: \\
for any $\psi\in \mathbb{C}^{A_0}$, 
\[ \psi^{\flat}(a) =\begin{cases} \psi(a) & \text{: $t(a)\in X$,}\\
-\psi(a) & \text{: $t(a)\in Y$,}
\end{cases} \]
and for any $f\in \mathbb{C}^{V_0}$, 
\[ f^{\flat}(u) =\begin{cases} f(u) & \text{: $u\in X$,}\\
-f(u) & \text{: $u\in Y$.}
\end{cases} \]
It holds that if we set $f(u):=\sum_{t(a)=u}\psi (a)$ for the bipartite case, then 
\begin{equation}
f^\flat (u)=\sum_{t(a)=u}\psi^\flat (a). 
\end{equation}

In the following,  let us show our previous results on the cases of $z=\pm 1$. 
\begin{theorem}[\cite{HSS1,HSS2}]\label{thm:varphi1}
Let us consider the two cases described above---namely, $z=1$ and  ``$z=-1$ and $G_0$ is bipartite". 
Let $\mathrm{j}\in  \mathbb{C}^{\tilde{A}}$ be the electric current function on $\tilde{G}$ with the following boundary condition on $\delta V_0=\{u_1,\dots,u_{|\delta V_0|}\}$: 
\[\mathrm{j}(e_k)=
\begin{cases}
\bs{\alpha}_{in}^{\flat}(u_k)-\mathrm{ave}(\bs{\alpha}_{in}^{\flat}),\;(k=1,\dots,|\delta V_0|) & \text{: $z=-1$ and $G$ is bipartite,} \\
\bs{\alpha}_{in}(u_k)-\mathrm{ave}(\bs{\alpha}_{in}),\;(k=1,\dots,|\delta V_0|) & \text{: $z=1$,}
\end{cases}
 \]
where $\mathrm{ave}(\bs{\alpha}_{in})=(\alpha_1+\cdots+\alpha_r)/{|\delta V_0|}$. 
Then for any connected graph $G_0$ we have
\begin{align*} 
\phi_{1}(a) &= 
\mathrm{j}(a)+\mathrm{ave}(\bs{\alpha}_{in}),
\end{align*}
while for any connected bipartite graph $G_0$ we have  
\begin{align*}
\phi_{-1}^{\flat}(a) &= 
\mathrm{j}(a)+\mathrm{ave}(\bs{\alpha}_{in}^{\flat}). \end{align*}
\end{theorem}
%

Hereinafter, we will reproduce the above theorem by taking the limit $z\to\pm 1$. 
As a by-product of reproducing this result, we find  
an interesting relation to the potential function of the quantum walk and that of the electric circuit as follows. 
\begin{theorem}\label{thm:continuous}
Set $\nu:[0,2\pi)\to \mathbb{C}^{V_0}$ by $\nu(\theta):=\nu_{e^{i\theta}}$. 
Then 
its derivatives at $\theta=0,\pi$ describe the potential functions of the electric circuit satisfying the following Poisson equations: for any connected graph $G_0$, 
\[ (P_0-I)(-i\nu'(0))=q_0,\;
 \]
and for any connected bipartite graph $G_0$ with the partite sets $V_0=X\sqcup Y$,
\[(P_0-I)(-i{\nu'(\pi)}^\flat)=q_0^\flat,\]
where 
\[q_0= D_0^{-1}\bs{\alpha}_{in}. \]
\end{theorem}
Now through the following proof, let us see how Theorem~\ref{thm:varphi1} is reproduced. 
\begin{proof}
Let $L_{e^{i\theta}}$ be denoted by $L(\theta)$ and the inverse matrix of $L(\theta)$ be denoted by $L(\theta)^{-1}$, which is a $|V_0|$-dimensional matrix. Recall that the potential function $\nu_z(u)$ is originally defined in (\ref{eq:nuz}) as the average of the stationary state $\phi_z(a)$'s over all the arcs whose terminal vertex is $u$. Then from Theorem~\ref{prop:PoiPoten}, the potential function 
$\nu(\theta):=\nu_{e^{i\theta}}$ on the $\epsilon$-neighborhood of $\theta=0,\pi$ for sufficiently small $\epsilon>0$, $B_{0,\epsilon}:=\{\theta \;:\; |\theta|<\epsilon\}$, $B_{\pi,\epsilon}:=\{\theta \;:\; |\theta-\pi|<\epsilon\}$, can be expressed as follows:
\[ \nu(\theta)=
i\sin \theta\;L(\theta)^{-1}\bs{\alpha}_{in}.  \]
Our first target is to show 
\[ \lim_{\theta\to 0}\nu(\theta) = \mathrm{ave}(\bs{\alpha}_{in})\bs{1}_{V_0} \]
and, if $G$ is bipartite, 
\[ \lim_{\theta\to \pi}\nu(\theta) = \mathrm{ave}(\bs{\alpha}^{\flat}_{in})\bs{1}^{\flat}_{V_0}, \]
because after 
the Kirchhoff current law to the electric current $\mathrm{j}$ for $\theta=0,\pi$ cases are applied, we obtain $(\nu(0))(u)=(1/\deg_{\tilde{G}}(u))\;\sum_{t(a)=u}\phi_1(a)=\mathrm{ave}(\bs{\alpha}_{in })$ and so on by Theorem~\ref{thm:varphi1}. 

We use the following lemma to express $L(\theta)^{-1}$. 
\begin{lemma}\label{lem:Kinverse}
We have
\begin{equation} 
\lim_{\theta\to 0} i\sin\theta L(\theta)^{-1}=\frac{1}{|\delta V_0|}J_{| V_0|}
\end{equation}
for any connected graph $G_0$, and  
\begin{equation}\label{eq:Kinversebipartite} 
\lim_{\theta\to \pi} i\sin\theta L(\theta)^{-1}=\frac{1}{|\delta V_0|}J_{|V_0|}^\flat 
\end{equation}
for any connected bipartite graph, where $J_{|V_0|}=\bs{1}_{V_0}\bs{1}_{V_0}^*$ and $J^\flat_{|V_0|}=\bs{1}_{V_0}^{\flat}{\bs{1}_{V_0}^\flat}^*$.
\end{lemma}
\begin{proof}
First let us see the case where $\theta\to 0$.  
Remark that from the definition of $L(\theta)$, obviously the matrix $L(\theta)$ is analytic on $\theta\in \mathbb{R}$. Indeed, the expansion of $L(\theta)$ around $\theta=0$ can be expressed by
\begin{align} 
L(\theta) &= (M_0-D_0)+\theta L^{(1)}+\theta^2 L^{(2)}+\theta^3 L(\theta)+\cdots \\
&= (M_0-D_0)+i\theta \Pi_{\delta V_0}+\theta^2\Gamma_\theta, \label{eq:Ktheta}
\end{align}
where $\Gamma_\theta$ is bounded, that is, there exists a constant $c$ such that  $||\Gamma_\theta||<c$ for any $|\theta|<1$. 
Let $\lambda(\theta)\in \sigma(L(\theta))$ be the eigenvalue whose absolute value is  closest to the value $0$ in all the eigenvalues of $L(\theta)$. 
If $\theta$ is sufficiently small, then this eigenvalue is the small perturbation of the maximal eigenvalue $\lambda_0=0$ of the Laplacian matrix $L_0=M_0-D_0$, which is simple. 
Therefore if $\theta\in B_{0,\epsilon}\setminus\{0\}$, then $\lambda(\theta)$ is simple and isolated.  
Then the resolvent of $L(\theta)$ is decomposed into the following partial fraction by \cite{Kato}:
\begin{equation}\label{eq:partialfraction}
    (L(\theta)-z)^{-1}=-\frac{1}{z-\lambda(\theta)}P_\theta \oplus A_\theta
\end{equation}
for any $z\notin \sigma(L(\theta))$. 
Here $P_\theta$ is the eigenprojection of the eigenvalue $\lambda(\theta)$ which can be expanded by $P_\theta=P^{(0)}+\theta P^{(1)}+\theta^2 P^{(2)}+\cdots$ and  
$||A_\theta||<c$ for sufficiently small $\theta$. 
Note that for any $\lambda\in \sigma(L(\theta))\setminus\{\lambda(\theta)\}$, it holds that $\lambda\not\to 0$ as $\theta\to 0$ because $\lambda(\theta)$ and $\lambda_0=0$ are simple and isolated. 
Since $0\notin \sigma(L(\theta))$, we can put $z=0$ in (\ref{eq:partialfraction}).
Then we have 
\begin{equation}\label{eq:sita}
    i\sin\theta\; L(\theta)^{-1} = \frac{i\sin\theta}{\lambda(\theta)}\;P_\theta \oplus i\sin \theta\; A_\theta.  
\end{equation}
Since $\lambda_0$ is simple, by \cite{Kato}, 
\[\lambda(\theta)=\lambda_0+\sum_{n=1}^{\infty}\theta^n \lambda^{(n)},\]
where 
\begin{equation}\label{eq:lambda1} \lambda_0=0,\;\lambda^{(1)}=\mathrm{tr}(L^{(1)} P^{(0)})=\mathrm{tr}(i\Pi_{\delta V_0}\frac{1}{|V_0|}J_{|V_0|})=i\frac{|\delta V_0|}{|V_0|}. 
\end{equation}
Note that $P_0$ is characterized as the projection to the constant function $\bs{1}_{V_0}$.  
Then we obtain
\begin{align}
    i\sin\theta L(\theta)^{-1}
    &= \frac{i\sin \theta}{\lambda_0+\theta\lambda^{(1)}+\theta^2 \lambda^{(2)}+O(\theta^3)} (P^{(0)}+\theta P^{(1)}+\cdots) \oplus i\sin\theta A_\theta \label{eq:limitK0'} \\
    &= \frac{i\theta+O(\theta^3)}{i\theta|\delta V_0|/|V_0|+O(\theta^2)} (P^{(0)}+\theta P^{(1)}+\cdots) \oplus i(\theta+O(\theta^3)) A_\theta \label{eq:limitK0} \\
    &\to 
    \frac{1}{|\delta V_0|/|V_0|}P^{(0)}
    = \frac{1}{|\delta V_0|}J_{|V_0|}\text{ ($\theta\to 0$).} \label{eq:limitK}
\end{align}
Next, let us consider the case where $\theta\to \pi$. We just replace $\theta$ with $\pi-\epsilon$ in the above discussion. 
The matrix $L(\theta)$ can be expanded around $\pi$ as follows:  
\[ L(\pi-\epsilon)=M_0+D_0+i\epsilon\Pi_{\delta V_0}+\epsilon^2 \tilde{\Gamma}_{\epsilon}, \]
which corresponds to (\ref{eq:Ktheta}), and $L(\pi-\epsilon)$ is a perturbation matrix of the {\it signless} Laplacian matrix $M_0+D_0$. 
We also obtain
\begin{align*} i\sin\epsilon L(\pi-\epsilon)^{-1} &= \frac{i\sin \epsilon}{\lambda(\pi-\epsilon)}P_{\pi-\epsilon}\oplus i\sin\epsilon A_{\pi-\epsilon}, 
\end{align*}
which corresponds to (\ref{eq:sita}).
Since $\lambda_\pi$ is simple, we have 
\[ \lambda(\pi-\epsilon) =\lambda(\pi)+\sum_{n=1}\epsilon^n{\tilde{\lambda}}^{(n)}, \]
where 
\[\lambda(\pi)=0,\; {\tilde{\lambda}}^{(1)}=\mathrm{tr}(i\Pi_{\delta V_0}\frac{1}{|V_0|}J_{|V_0|}^{\flat})=i\frac{|\delta V_0|}{|V_0|} \]
because ${\tilde{P}}^{(0)}=(1/|V_0|)\;\bs{1}_{V_0}^{\flat}{\bs{1}_{V_0}^\flat}^*$.
Then by an approach similar to that used in (\ref{eq:limitK0'})--(\ref{eq:limitK}), we obtain 
\[ \lim_{\theta\to\pi}i\sin\theta L(\theta)^{-1}=-\frac{1}{|\delta V_0|}J_{|V_0|}^{\flat}. \]
This completes the proof of Lemma~\ref{lem:Kinverse}. 
\end{proof}
From this lemma, we obtain 
\begin{equation}\label{eq:nu0}
\lim_{\theta\to 0}\nu(\theta)=\lim_{\theta\to 0}i\theta L(\theta)^{-1}\bs{\alpha}_{in}=\frac{1}{|\delta V_0|}J_{|V_0|}\bs{\alpha}_{in}=\mathrm{ave}(\alpha)\bs{1}_{V_0}=\nu(0), 
\end{equation}
which shows the continuity of $\nu(\theta)$ at $\theta=0$. 
Next, let us see the differentiability of $\nu(\theta)$ at $\theta=0$. 
Note that the definition of $\nu$ and  (\ref{eq:lambda1}) imply 
\begin{align*}
\nu(0) &= \mathrm{ave}(\bs{\alpha}_{in})\bs{1}_{V_0} 
= \frac{1}{|\delta V_0|/|V_0|} P^{(0)}\bs{\alpha}_{in} \\
&= \frac{i}{\lambda^{(1)}}P^{(0)}\bs{\alpha}_{in}. 
\end{align*}
From (\ref{eq:limitK0'}), we have 
\begin{multline*}
    \frac{\nu(\theta)-\nu(0)}{\theta}\\ =\frac{1}{\theta}\left\{ \left(\frac{i\sin \theta}{\lambda_0+\theta\lambda^{(1)}+\theta^2 \lambda^{(2)}+O(\theta^3)} (P^{(0)}+\theta P^{(1)}+\cdots) \oplus i\sin\theta A_\theta\right)\bs{\alpha}_{in}-\frac{i}{\lambda^{(1)}}P^{(0)}\bs{\alpha}_{in}\right\}\\
    \to 
    -i \frac{\lambda^{(2)}}{\{\lambda^{(1)}\}^2}P^{(0)}\bs{\alpha}_{in}+i\frac{1}{\lambda^{(1)}}P^{(1)}\bs{\alpha}_{in}+iA_0\bs{\alpha}_{in},\;\;(\theta\to 0)
\end{multline*}
which shows the differentiability of $\nu(\theta)$ at $\theta=0$.
We put $d\nu(\theta)/d\theta|_{\theta=0}=:\nu'_1$ and the difference between $\nu(0)$ and $\nu(\theta)$ by 
\[ \Delta_\theta:=\nu(\theta)-\nu(0). \]
The LHS of (\ref{eq:startingpoint}) in Theorem~\ref{prop:PoiPoten} can be expanded by using (\ref{eq:Ktheta}) as follows:  
\begin{align*}
    L(\theta)\nu(\theta) &= \left(D_0(P_0-I)+i\theta \Pi_{\delta V_0}\right) (\nu_1+\Delta_\theta) +O(\theta^2) \\
    &= i\theta \Pi_{\delta V_0}\nu_1+D_0(P_0-I)\Delta_\theta+i\theta \Pi_{\delta V_0}\Delta_\theta+O(\theta^2).
\end{align*}
On the other hand, the RHS of (\ref{eq:startingpoint}) in Theorem~\ref{prop:PoiPoten} can be expanded by
\[ j_-(e^{i\theta})\bs{\alpha}_{in}=i\theta\bs{\alpha}_{in} +O(\theta^3), \]
where $j_-(z)$ is defined in Definition~\ref{def:jz}. 
Combining the above, we have 
\[ (P_0-I)\frac{\Delta_\theta}{i\theta}=D_0^{-1}(\bs{\alpha}_{in}-\Pi_{\delta V_0}\nu(0))-D_0^{-1}\Pi_{\delta V_0}\Delta_\theta+O(\theta). \]
Taking the limit of both sides where $\theta\to 0$, we obtain that $\nu'(0)$ satisfies 
the following Poisson equation of the electric circuit:
\begin{equation}\label{eq:poissonEQ} (P_0-I)(-i\nu'(0))=q_0. \end{equation}
Here $q_0=D_0^{-1}(\bs{\alpha}_{in}-\Pi_{\delta V_0}\nu(0))$. It is enough to check that 
\[ \langle f,q_0 \rangle_\pi:=\sum_{u\in V_0} \bar{f}(u)q_0(u) d(u)=0.  \]
However, since $f$ is a constant function, 
it is easy to check that the above equation holds by the definition of $q_0$ and (\ref{eq:nu0}). 
Therefore the differential of $\nu(\theta)$ at $\theta=0$, $\nu_1'\in\mathbb{C}^{V_0}$, describes the current flow $\mathrm{j}$ which satisfies the Kirchhoff current and voltage laws: rewriting $\nu_1,\nu_1'\in\mathbb{C}^{V_0}$ by $\nu_1=\nu(0)$, $\nu_1'=\nu'(0)$ again, we have 
\begin{equation}\label{eq:potential} \mathrm{j}(a) = (-i)\nu_1'(t(a))-(-i)\nu_1'(o(a)) 
\end{equation}
for any $a\in A_0$. 
Finally, let us feed this back to the stationary state using (\ref{eq:startingpoint2}) in Theorem~\ref{prop:PoiPoten}. 
The LHS of (\ref{eq:startingpoint2}) is expanded to $i\theta \phi_{e^{i\theta}}(a)+O(\theta^3)$ for any $a\in A_0$. 
The RHS of (\ref{eq:startingpoint2}) is expanded to 
\begin{multline*} (1+i\theta)(\;\nu_1(t(a))+\Delta_\theta(t(a))\;)-(\;\nu_1(o(a))+\Delta_\theta(o(a))\;) +O(\theta^2)\\
=  i\theta \nu_1(t(a)) +\Delta_\theta(t(a))-\Delta_\theta(o(a))+O(\theta^2)
\end{multline*}
for any $a\in A_0$. 
Then we have 
\begin{equation*}
\lim_{\theta\to 0}\phi_{e^{i\theta}}(a) 
= \mathrm{ave}(\bs{\alpha}_{in}) + \mathrm{j}(a)
\end{equation*}
by (\ref{eq:nu0}) and (\ref{eq:potential}). This is nothing but $\phi_1$ in Theorem~\ref{thm:varphi1}, which completes the proof for the $z=1$ case. 

For $z=-1$ and $G_0$ is bipartite, 
in a similar fashion, starting from (\ref{eq:Kinversebipartite}) in  Lemma~\ref{lem:Kinverse} and replacing the input parameter $\theta$ with $\pi-\epsilon$,  we can see the continuity and differentiablility of $\nu(\theta)$ at $\theta=\pi$. The derivative at $\theta=\pi$; $\nu_{-1}'$, satisfies 
\[(P_0+1)(i\nu_{-1}')= q_0, \]
which corresponds to (\ref{eq:poissonEQ}).
However, note that it holds that 
\[ \left((P_0+1)f\right)^\flat=-(P_0-1)f^\flat \]
for any $f\in\mathbb{C}^{V_0}$ if $G_0$ is bipartite. Then we have
\[ (P_0-1)(-i\nu_1')^\flat=q_0^\flat. \]
Thus using the same approach as in the $z=1$ case, we see that $(-i\nu_1')^\flat$ is the potential function of the electric circuit and $\lim_{\theta\to\pi}\phi^\flat_{e^{i\theta}}(a)=\mathrm{ave}(\bs{\alpha}_{in})+\mathrm{j}(a)$. 
This completes the proof of Theorem~\ref{thm:continuous}.
\end{proof}

From this theorem, we also obtain the scattering matrix at $z=\pm 1$. 
\begin{corollary}[\cite{HSS1,HSS2}]\label{cor:scatteringpm1}
The scattering matrices for $z=\pm 1$ are expressed by 
\[ 
S_1 = \mathrm{Gr}(|\delta V_0|),
\]
for any connected graph $G_0$, while 
\[
S_{-1}= 
\begin{cases}
I & \text{: $G_0$ is non-bipartite,} \\
\begin{bmatrix}I_{X} & 0 \\ 0 & -I_{Y}\end{bmatrix} \mathrm{Gr}(|\delta V_0|) \begin{bmatrix}I_{X} & 0 \\ 0 & -I_{Y}\end{bmatrix} & \text{: $G_0$ is bipartite.}
\end{cases}
\]
\end{corollary}
\begin{proof}
Let us consider the case for $z=1$.
By Theorem~\ref{cor:scattering}, we have 
\begin{equation}\label{eq:scattering_perturbed}
S_{e^{i\theta}}=e^{i\theta}\;\chi_{\delta {V_0}}\; \left\{2j_-(z)L(e^{i\theta};G_0;\delta V_0)^{-1}-I_{V_0}\right\}\;\chi_{\delta V_0}^* 
\end{equation}
for sufficiently small $\theta$. 
By Lemma~\ref{lem:Kinverse}, we have 
\[ \lim_{\theta\to 0}S(e^{i\theta};G_0;\delta V_0)=\chi_{\delta V_0}(\frac{2}{|\delta V_0|}J_{|V_0|}-I_{V_0})\chi_{\delta V_0}^*
=\frac{2}{|\delta V_0|}J_{|\delta V_0|}-I_{\delta V_0} \]
by (\ref{eq:limitK}). 
The case for $z=-1$ and ``$G_0$ is bipartite" can also be obtained in a similar way. 
Finally, in the case for $z=-1$ and ``$G_0$ is non-bipartite", noting that the signless Laplacian $L_{-1}$ becomes invertible, we obtain the desired conclusion. 
\end{proof}
\begin{remark}
The comfortabilities for $z=\pm 1$ are characterized by some graph geometries in \cite{HSS2}.  
The example can be seen for the complete graph case with arbitrary frequency in Section~5.2. 
\end{remark}

%
\section{Example: Complete graph case}\label{sect:example}
In this section, we consider scattering matrix and comfortability in the case of the complete graph with the vertex number $N$ and the boundary number $\ell$. The inflow penetrates the internal graph from a fixed vertex, say $v_1$. See Fig.~\ref{fig:zu}. 

The complete graph is simple but a nice example in that the transition matrix $P_0$ has an eigenvector which has no overlap to $\delta V$ if $1\leq|\delta V|<N-1$.
This means $j_+^{-1}(\sigma_{per})\neq \emptyset$; that is, $\mathbb{B}_*\cap \delta \mathbb{D}\neq \emptyset$. 
Let us see that as follows. 
The eigensystem of $P_0$ is described by  $\mathrm{spec}(P_0)=\{1,-1/(N-1)\}$
and 
\[\ker(P_0-1)=\mathrm{span}\{[1,\dots,1]^\top\}, \]
\[\ker\left(P_0+\frac{1}{N-1}\right)=\mathrm{span}\{[1,-1,0,\dots,0]^\top,\;[0,1,-1,\dots,0]^\top,\dots,[0,\dots,1,-1]^\top\}.\]
Thus the latter eigenspace is spanned by functions which have a finite range support with $2$ vertices. This means that there exists at least one eigenvector which has no overlap to $\delta V$ if $|\delta V|<N-1$.  
Then for the complete graph, 
$\mathbb{B}_*\cap \delta \mathbb{D}=\{ 1, e^{\pm i \theta_*}\}$ if $|\delta V|<N-1$, 
where $\cos\theta_*=-1/(N-1)$. 
We will see a special response of the scattering and also comfortability at such input $z\in \mathbb{B}_*$.  
\begin{figure}[hbtp]
    \centering
    \includegraphics[keepaspectratio, width=100mm]{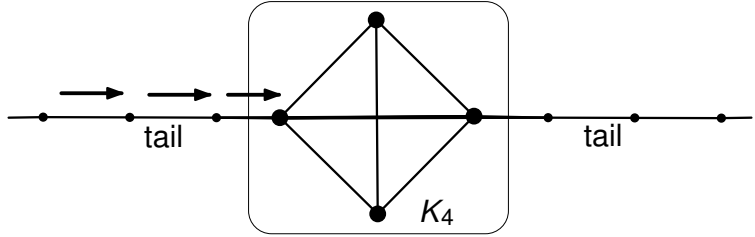}
    \caption{The situation for the complete graph case ($N=4$, $\ell=2$, $\bs{\alpha}_{in}=[1,0,0,0]^\top$). }
    \label{fig:zu}
\end{figure}
\subsection{Scattering matrix}
Put $S(\theta):=S_{e^{i\theta}}$. 
For the complete graph with $2$-boundary, the scattering is the perfect transmission for $z=1$, while it is the perfect reflection for $z=-1$ \cite{HSS1}.  
From Theorem~\ref{cor:scatteringpm1}, we can easily confirm  $S(0)=\sigma_X$ while $S(\pi)=I$ which is consistent with \cite{HSS1}. To reveal the  continuous connection between them, let us set the transmitting rate $t(\theta):=|(S(\theta))_{2,1}|^2=|(S(\theta))_{1,2}|^2$. Note that $t(0)=1$, $t(\pi)=t(\pm\theta_*)=0$. 
Where $\theta_*=\arccos(-1/(N-1))$. 
As a result of this subsection, we can continuously connect them: Figure~\ref{Fig:scattering} for $N=4$ shows the transmitting rate. 
\begin{figure}[hbtp]
    \centering
    \includegraphics[keepaspectratio, width=80mm]{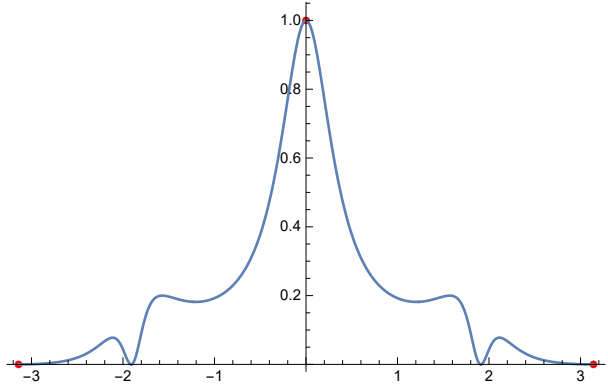}
    \caption{The scattering of quantum walk on the complete graph with $4$-vertex and $\ell=2$, $\bs{\alpha}_{in}=[1,0]^\top$:  The horizontal and vertical lines are the frequency of the inflow $\theta\in[-\pi,\pi]$ and the transmitting rate in the case of $N=4$. The red points have been obtained in \cite{HSS1}. We continuously connect them in this paper. }
    \label{Fig:scattering}
\end{figure}
We obtain an explicit expression for the inverse of $L(\theta)^{-1}$ as follows. 
\begin{lemma}\label{lem:CompleteGraphLInv}
Let the internal graph $G_0$ be the complete graph with the $N$-vertex and $\ell$-boundary. 
We set $m=N-\ell$. 
Let $\alpha=-1-(N-1)\cos\theta$, $\beta=\alpha+i\sin\theta$. 
Then the inverse of $L(\theta)$ is expressed by  
\begin{equation}
     (L(\theta)^{-1})_{i,j}=\frac{1}{\beta\alpha(\beta\alpha+\ell \alpha+m\beta)}
\begin{cases}
\beta\alpha^2+(\ell-1)\alpha^2+m\beta\alpha & \text{:  $1\leq i,j\leq \ell$, $i=j$,}\\
-\alpha^2 & \text{:  $1\leq i,j\leq \ell$, $i\neq j$,}\\
\alpha\beta^2+(m-1)\beta^2+\ell\beta\alpha & \text{: $\ell< i,j\leq \ell+m$, $i=j$,}\\
-\beta^2 & \text{: $\ell< i,j\leq \ell+m$, $i\neq j$,}\\
-\beta\alpha & \text{: otherwise}
\end{cases} 
\end{equation}
for $1\leq i,j\leq N$. 
\end{lemma}
\begin{proof}
If $L(\theta)$ is invertible, then the scattering matrix can be described by
\[ e^{-i\theta}S(\theta)=2i\sin\theta \chi_{\delta V_0}L(\theta)^{-1}\chi_{\delta V_0}^* -I_{\delta V_0}. \]
Since the adjacency matrix $M_0$ is of the form $J-I$, $L(\theta)$ can be expressed by 
\begin{align} L(\theta) &= J-(1+(\ell+m-1)\cos\theta)I+i\sin\theta \Pi_{\delta V_0}  \\
&= [\; \bs{u}+\beta\bs{e}_1\;|\;\cdots\;|\;\bs{u}+\beta\bs{e}_\ell\;|\;\bs{u}+\alpha \bs{e}_1'\;|\;\cdots\;|\;\bs{u}+\alpha\bs{e}_{m}' \;], \label{eq:expression}
\end{align}
where $m=N-\ell$, $\bs{u}=[1,\dots,1]^\top$, $\bs{e}_j$ is the $j$-th standard basis, $\bs{e}_j'=\bs{e}_{\ell+j}$ and 
$\alpha=-1-(\ell+m-1)\cos\theta$, $\beta=\alpha +i\sin\theta$. 
Note that 
\[ L(\theta)^{-1}=\tilde{L}_\theta/\det (L(\theta)). \]
Here $\tilde{L}_\theta$ is the adjugate matrix of $L(\theta)$. 
From the expression of $L(\theta)$ in (\ref{eq:expression}),  we can expand the determinant. Note that if the vector $\bs{u}$ appears more than two times at the column vectors, then such a determinant is $0$. Then it holds that  
\begin{multline}
    \det(L(\theta))= \det\;[\;\beta \bs{e}_1\;|\;\cdots\;|\;\beta\bs{e}_\ell\;|\;\alpha \bs{e}_1'\;|\;\cdots\;|\;\alpha\bs{e}_m'\;] \\
    + \sum_{j=1}^\ell \det [\;\beta \bs{e}_1\;|\;\cdots\;|\;\underset{ j}{\bs{u}}\;|\;\cdots\;|\;\beta\bs{e}_{\ell}\;|\;\alpha \bs{e}_1'\;|\;\cdots\;|\;\alpha \bs{e}_m'\;] \\
    +\sum_{j=1}^m \det[\; \beta \bs{e}_1\;|\;\cdots\;|\;\beta\bs{e}_{\ell}\;|\;\alpha \bs{e}_1'\;|\;\cdots\;|\;\underset{j}{\bs{u}}\;|\;\cdots\;|\;\alpha \bs{e}_m' \;].
\end{multline}
Then we have 
\[ \det(L(\theta))=\beta^{\ell}\alpha^m+\ell\beta^{\ell-1}\alpha^m+m\beta^{\ell}\alpha^{m-1}. \]
In a similar fashion, we obtain
\[(\tilde{L}_\theta)_{i,j}=\begin{cases}  \beta^{\ell-1}\alpha^m+(\ell-1)\beta^{\ell-2}\alpha^m+m\beta^{\ell-1}\alpha^{m-1} & \text{: $1\leq i,j\leq \ell$, $i=j$, }\\ 
-\beta^{\ell-2}\alpha^m & \text{: $1\leq i,j\leq \ell$, $i\neq j$, }\\
\alpha^{m-1}\beta^\ell+(m-1)\alpha^{m-2}\beta^\ell+\ell\alpha^{m-1}\beta^{\ell-1} & \text{: $\ell< i,j\leq N$, $i=j$, }\\ 
-\alpha^{m-2}\beta^\ell & \text{: $\ell< i,j\leq N$, $i\neq j$, }\\
-\alpha^{m-1}\beta^{\ell-1} & \text{: otherwise.}
\end{cases}\]
Then we obtain the desired conclusion.
\end{proof}
Let $S_\theta$ be the scattering matrix on the surface. Then we obtain the following expression for $S_\theta$.
\begin{proposition}
Let the frequency of the input  be $z=e^{i\theta}$. 
Let $K_{\ell+m}$ be the complete graph with $\ell+m$ vertices and $\ell$ boundaries. 
Then the scattering matrix on this graph is described as follows: 
\[ e^{-i\theta}S_\theta = 2ix\sin\theta J_\ell+(2iy \sin\theta-1)I_\ell  \]
where 
\begin{equation}\label{eq:x}
x=-\frac{\alpha}{\beta}\frac{1}{\alpha \beta + \alpha \ell + \beta m},\;y=\frac{1}{\beta}.    
\end{equation} 
\end{proposition}
\begin{proof}
If $L(\theta)$ is invertible, then the scattering matrix can be described by
\[ e^{-i\theta}S_\theta=2i\sin\theta \chi_{\delta V_0}\;L(\theta)^{-1}\;\chi_{\delta V_0}^* -I_{\delta V_0}. \]
Inserting the expression of the inverse of $L(\theta)$ in Lemma~\ref{lem:Kinverse} into the above, we obtain the desired conclusion. 
\end{proof}
The expression for the coefficients of $J_\ell$ and $I_\ell$ can be described as 
\begin{align*}
    2ix\sin\theta &= \frac{1+(N-1)c}{1+(N-1)c-is}\;\frac{2}{\ell c-t'\{(m-1)s+i(N-1)(1+(N-1)c)\}}\\
    2iy\sin\theta-1 &= \frac{1+(N-1) c+is}{-1-(N-1) c+is}, 
\end{align*}
where $c=\cos\theta$, $s=\sin\theta$, and $t'=\tan(\theta/2)$. 
From the above expressions, the condition under which the RHS becomes the diagonal matrix is  $\theta=0,\pi$ or $\alpha=0$. Note that $\alpha=0$ if and only if $\theta=\pm \theta_*$, where $\theta_*=\mathrm{Arccos}(-1/(N-1))$. 
By the  L'Hôpital's rule, if $\theta=0$, \[\lim_{\theta\to 0}S_\theta= \frac{2}{\ell}J_{\ell}-I_{\ell}, \]
which is consistent with the result on \cite{HSS1}, while for $\theta=\pi$ and $\pm\theta_*$, 
\[ S_\pi=I_{\ell},\;S_{\pm \theta_*}=e^{\pm i \theta_*} I_{\ell}. \]
Thus we obtain the following corollary. 
The perfect reflection happens for ``$\theta=\pm\theta_*$ and $m\geq 1$" or ``$\theta=\pm \pi$". 
\begin{corollary}
Assume $m=1,2,\dots , N-2$.
$\theta=\pi$ or $\theta=\pm \theta_*$ if and only if 
every quantum walker goes out from the same place where it came in; that is, the perfect reflection.
\end{corollary}
\begin{proof}
From (\ref{eq:x}), it is enough to clarify when $x=0$.
\end{proof}
%
\begin{remark}
Let us assume that $m=0$, equivalently, $\ell = N$. Then the perfect
reflection occurs if and only if $\theta = \pi$.
Moreover , for $\theta = \pm\theta_{*}$, we can easily obtain
\[
S_{\pm\theta_{*}} = e^{i\pm\theta_{*}} (
-\frac{2}{N+i\sin(\pm\theta_{*})}J_{\ell} + I_{\ell}).
\]
Finally, let us assume that $m=N-1$, 
equivalently $\ell=1$. 
Under this condition, 
we can consider that
the perfect reflection 
always occurs
since only one
tail exists and the inflow penetrates 
along it.
Thus  we can calculate a scalar  $S_{\theta}$ by using (\ref{eq:x}) 
with 
$
\sin^{2}\theta=
(N+\alpha)(N-2-\alpha)/{(N-1)^{2}}
$
for 
$e^{-i\theta}S_{\theta} = 2i(x+y)\sin\theta-1$. As a result, we have
$S_{\theta} = e^{i(\theta +\tilde{\theta})}$, where
\[
\cos \tilde{\theta} =
\frac{
(\alpha + N-1)^{2}(N-2-\alpha) - \alpha^{2}(\alpha + N)(N-1)^{2}
}
{
(\alpha + N-1)^{2} (N-2-\alpha) + \alpha^{2}(\alpha+N)(N-1)^{2}
}
\]
and
\[
\sin \tilde{\theta} =
\frac{
2\alpha (\alpha + N-1) (N-1)^{2}\sin\theta
}
{
(\alpha +N-1)^{2}(N-2-\alpha) + \alpha^{2}(\alpha+N)(N-1)^{2}
}.
\]
Recall $\alpha = -1 -(N-1)\cos\theta$.
\end{remark}
For example, let us set $G_{0}=K_{4}$, which 
is the complete graph of four vertices. 
We illustrate the transmitting probability of $\ell$-tails ($\ell=1,2,3, 4$)
for $\theta\in [-\pi, \pi]$ in 
Figure~\ref{Fig:scattering1}. 
Here the transmitting probability 
is defined by $1-(\text{reflection probability})^{2}$.
Perfect reflection can be observed at the point where the curve touches the horizontal line.

%

Let us compare the two walkers along the tails who were reflected perfectly by the internal graph. 
If $\theta=\pi$, the quantum walker of the outflow is completely same as when it was the inflow, while if $\theta=\pm \theta_*$, the appearance of the quantum walker of the outflow is changed from that of when it was the inflow because of the twisted term $e^{\pm i\theta_*}$. 
Then by observing the phase of the quantum walker coming back to the same tail, we can detect whether this quantum walker felt comfortable in the internal graph or not, because $\mathcal{E}_{\pi}\in O(1/N)$, while  $\mathcal{E}_{\pm\theta_*}\in O(N)$; see the next subsection for more detail. 
\begin{figure}[hbtp]
    \centering
    \includegraphics[keepaspectratio, width=100mm]{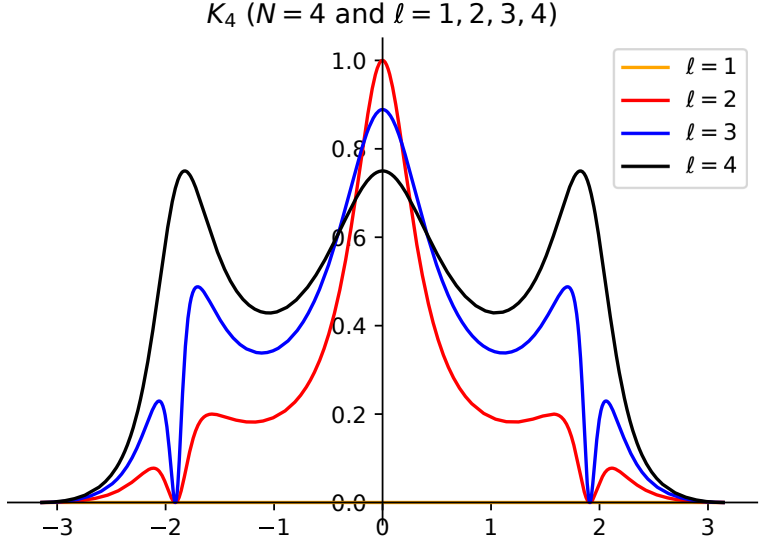}
    \caption{The transmitting probability on the complete graph with $4$ vertices:  The horizontal and vertical lines are the frequency of the inflow $\theta\in[-\pi,\pi]$ and the transmitting rate to all  tails except the receiver's tail for each $\ell=1,2,3,4$.} 
    \label{Fig:scattering1}
\end{figure}
\subsection{Comfortability}
Let us consider the complete graph with the $N$-vertex and $|\delta V_0|=\ell$ ($1\leq \ell\leq N$) with $\bs{\alpha}_{in}=[1,0,\dots,0]^\top$. 
If $N=4$, $\ell=2$ and $\bs{\alpha}_{in}=[1,0,0,0]^\top\in\mathbb{C}^{4}$, the comfortabilities for $z=1$ and $z=-1$ are computed as $13/8$ and $5/12$, respectively~\cite{HSS1}. 
Let the comfortability for $z=e^{i\theta}$ be denoted by 
$\mathcal{E}(\theta):=\mathcal{E}_{e^{i\theta}}$. 
In this section, we are interested in connecting $\mathcal{E}(0)$ to $\mathcal{E}(\pi)$ by continuously changing the frequency of the inflow from $\theta=0$ to $\theta=\pi$. 
See Figure~\ref{fig:zu}.

Using the specialty of the complete graph, $M_0=J-I$, $D_0=(N-1)I$,  (\ref{eq:ComK}) in  Theorem~\ref{cor:comf} is further reduced to \begin{equation}\label{eq:1}
    \mathcal{E}(\theta)=(N-1+\cos\theta)||L(\theta)^{-1}\bs{\alpha}_{in}||^2-\cos\theta\langle L(\theta)^{-1}\bs{\alpha}_{in},JL(\theta)^{-1}\bs{\alpha}_{in} \rangle.
\end{equation}
Here $\bs{\alpha}_{in}=[1,0,\dots,0]^\top\in\mathbb{C}^N$. 

From Lemma~\ref{lem:CompleteGraphLInv}, $L(\theta)^{-1}\bs{\alpha}_{in}$ can be expressed by
\[ L(\theta)^{-1}\bs{\alpha}_{in}=[p,\overbrace{q,\dots,q}^{\ell-1},\overbrace{r,\dots,r}^{m}]^\top, \]
where 
\begin{align*}
    p =\frac{\beta \alpha^2+(\ell-1)\alpha^2+m\beta\alpha}{\alpha\beta(\alpha \beta +\ell\alpha+m\beta)},\;
    q= \frac{-\alpha^2}{\alpha\beta(\alpha \beta +\ell\alpha+m\beta)}, \;
    r= \frac{-\alpha\beta}{\alpha\beta(\alpha \beta +\ell\alpha+m\beta)}. 
\end{align*}
Then inserting the above into (\ref{eq:1}), we show that the comfortability of the quantum walk on the complete graph with the  $N$-vertex, $\ell$-boundary for the $\bs{\alpha}=[1,0,\dots,0]^\top$ inflow with the frequency $z=e^{i\theta}$ can be expressed by 
\begin{equation}\label{eq:complete}
    \mathcal{E}(\theta)=(N-1+\cos\theta)\;\{\;|p|^2+(\ell-1)|q|^2+m |r|^2\;\}-\cos\theta\; |p+(\ell-1)q+m r|^2.
\end{equation}
 Figure~\ref{fig:comf} depicts the comfortability $\mathcal{E}(\theta)$ for $N=4$ and $\ell=2$. 
Starting from this expression, let us observe some interesting aspects of quantum walks through the comfortability of the complete graph. 
\begin{figure}[hbtp]
    \centering
    \includegraphics[keepaspectratio, width=100mm]{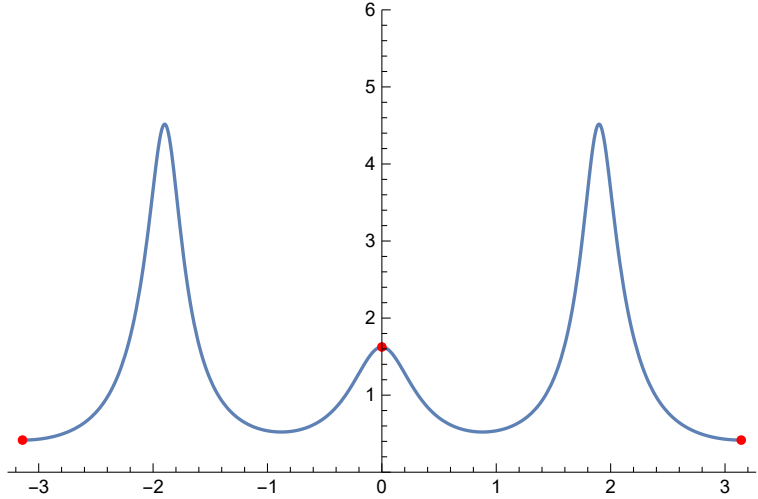}
    \caption{The comfortability of quantum walk on the complete graph with $4$-vertex with $2$ boundaries:  The horizontal and vertical lines are frequency of the inflow and the comfortability. The red points at $\theta=0$ and $\theta=\pm \pi$ are obtained by \cite{HSS1}. In this paper, we continuously connect them. }
    \label{fig:comf}
\end{figure}

After some algebra, the expression for the comfortability (\ref{eq:complete}) can be rewritten as  
\begin{equation}\label{eq:comf2}
    \mathcal{E}(\theta)=\frac{(N-1) (\alpha^2f_1+m(2\alpha+m+1)f_2)}{(N(N-2)(\alpha^2+1)-2\alpha)(\alpha^2f_3-m (\alpha-N+2)(2\alpha+m))},
\end{equation}
where 
\begin{align*}
    f_1 &:= f_1(\alpha;N)=(N-1)^2(\alpha(N^2-4N+2)+(N-2)(N^2-N+1)), \\
    f_2 &= f_2(\alpha;N)=(\alpha-N+2)(\alpha-N^2+2 N), \\
    f_3 &= f_3(\alpha;N)=N(N-2)\alpha+N^3-2N^2+2N-2.
\end{align*}
Recall that $\alpha=-1-(N-1)\cos\theta$. Then the range of $\alpha$ is $[-N,N-2]$. 
Note that for the complete graph case, $L(\theta)$ does not have the inverse if and only if $\theta=0, \pm \theta_*$ but the form (\ref{eq:comf2}) has no singularities in $[-N,N-2]$. This is consistent with Theorem~\ref{thm:stationarystate}. 
From this expression, for example, we give some comfortability at special frequencies of inflow: 
$\theta=0(\Leftrightarrow \alpha=N-2)$, $\theta=\pm\pi/2(\Leftrightarrow \alpha=-1)$, $\theta=\pm\theta_*(\Leftrightarrow \alpha=0)$ and $\theta=\pi(\Leftrightarrow \alpha=-N)$. 
\begin{align*}
    &\mathcal{E}(0) =\frac{N(N-1)}{2\ell^2}+\frac{\ell-1}{\ell N}; \\
    &\mathcal{E}({\pm\pi/2}) = \frac{(N-1)\left\{ N^2-3N+4+m(m-1) \right\}}{2 \left\{ (N-1)^2+(m-1)^2 \right\}}; \\
    &\mathcal{E}({\pm \theta_*}) = \begin{cases}
    \frac{m+1}{m}(N-1) & \text{: $m\neq 0$;}\\
    \\
    \frac{(N-1)^3(N^2-N+1)}{N(N^3-2N^2+2N-2)} & \text{: $m=0$;}
    \end{cases}  \\
    &\mathcal{E}({\pi}) = \frac{2N-3}{2(N-1)(N-2)}.
\end{align*}
Here the scattering for the frequency $\theta=0(\Leftrightarrow\alpha=N-2)$ is described by the Grover matrix, while those for $\theta=\pi(\Leftrightarrow\alpha=-N)$ and $\theta=\theta_*(\Leftrightarrow\alpha=0)$ are the perfect reflection. Their $\theta$'s are included in $\mathbb{B}_*$. 
The reason for our choice of the comfortability at $\theta=\pm\pi/2$ comes from the following curious phenomenon. The definition of $\alpha$ implies $\lim_{N\to\infty}\theta_*=\pi/2$ while 
\[ \lim_{N\to\infty}\frac{\mathcal{E}({\pi/2})}{N}=\frac{1}{2}<\frac{m+1}{m}=\lim_{N\to\infty}\frac{\mathcal{E}(\theta_*)}{N}. \] 
We discuss this after (\ref{eq:comfasymp}). 
\\

\noindent{\bf (1) Is the comfortability is monotone decreasing with respect to the number of boundaries ? }\\
From the classical point of view, the larger the number of drains is, the harder the water pools in the tank in the stationary state, because there are many escape routes for the water. 
Let us consider the quantum walk case. 
For $\theta=0$, the comfortability is monotone decreasing with respect to $\ell$. 
This means that if $\theta=0$, then the larger the number of boundaries, the more the quantum walker feels uncomfortable, which seems reasonably intuitive, because there are many opportunities for a quantum walker in the complete graph to go to the outside. 
On the other hand, for $\theta=\theta_*$, the comfortability for $m$'s have the following relation:
\[ \mathcal{E}(\theta_*)|_{m=1}>\mathcal{E}(\theta_*)|_{m=2}>\mathcal{E}(\theta_*)|_{m=3}>\cdots>
\mathcal{E}(\theta_*)|_{m=N-1}>
\mathcal{E}(\theta_*)|_{m=0}. \]
Note that $m=N-\ell$. 
Thus against the intuitive explanation of $\mathcal{E}_{0}$, the more we remove the  boundaries from the complete graph, the more the quantum walker feels uncomfortable if there is at least one vertex which is not the boundary.
Thus the situation in which every vertex  except one has the tail is the best for the comfortability of the quantum walker if $\theta=\theta_*$. 
However, if we add the ``extra" tail to this best situation for $\theta=\theta_*$, then the comfortability declines to its nadir.
When $\theta=\pi$, the comfortability is independent of the number of boundaries which also runs counter to intuition. We will show that the worst frequency for the comfortability is $\theta=\pi$.   
\\

\noindent{\bf (2) Combinatorial quantity obtained by the comfortability at $\theta=0$}\\
Recall that the comfortability at $\theta=0$ can be expressed by 
\begin{equation}\label{eq:0energy} \mathcal{E}({0})=\frac{N(N-1)}{2\ell^2}+\frac{\ell-1}{\ell N}.  
\end{equation}
The comfortability of the general graph $G=(V,E)$ for $\ell=2$ and $\theta=0$ is  expressed by the following graph geometry:  let $\chi_1$ be the number of spanning trees and $\chi_2$ be the number of spanning forests with two connected components---one including $u_1$ and the other including $u_N$, where $u_1$ is the vertex connecting to the tails of the inflow and $u_N$ is the vertex connecting to the other tail. According to \cite{HSS1}, we have
 \begin{equation}\label{eq:energy}
 \mathcal{E}(0)=\frac{1}{4}\;\left(\frac{\chi_2}{\chi_1}+|E|\;\right)  \end{equation}
for the general graph $G=(V,E)$. 
It is well known that the number of spanning trees of the complete graph with $N$ vertices is $N^{N-2}$, which is known as Cayley's formula. Thus (\ref{eq:0energy}) and (\ref{eq:energy}) lead to the following combinatorial value of the complete graph: 
\begin{equation}\label{eq:forest}
\chi_2 =2N^{N-3} 
\end{equation}
This equality is induced by our studies on quantum walks that make a detour, but it can also be explained directly as follows. 
Let $K_{x,y}:=(T_x,T_y)$ be the spanning forest of the complete graph with two subtrees $T_x$ and $T_y$. Here $T_z$ ($z\in\{x,y\}$) includes vertex $z$, and may be isolated. 
There are $\#V(T_x)\times \#V(T_y)$ pairs of $(x,y)$ inducing the subgraph $(T_x,T_y)$. 
On the other hand, there are also $\#V(T_x)\times \#V(T_y)$ spanning trees inducing the subgraph $(T_x,T_y)$ by eliminating a single appropriate edge. 
Let us consider the multiple set $\cup_{(x,y)}K_{x,y}$. From the above observation, we have
\[ \bigcup_{(x,y)}K_{x,y}=\bigcup_{T:\text{spanning tree}}\;\;\bigcup_{e\in E(T)}(T\setminus\{e\}). \]
The cardinality of RHS is described by $\binom{N}{2}\chi_2$, while the cardinality of LHS is described by $(N-1)\chi_1$. Then we have (\ref{eq:forest}). 
This equality also induces the following formula: 
\[\sum_{k=1}^{N-1}\binom{N-2}{k-1}k^{k-2}(N-k)^{N-k-2}=2N^{N-3}. \] 
\\

\noindent{{\bf (3) What is the worst situation of the boundary in terms of the comfortability?}} \\
From (\ref{eq:comf2}), we can show the following proposition. 
\begin{proposition}\label{prop:Kenergy}
Let us fix the parameter $\theta$. Then we have 
\[ \mathcal{E}(\theta)|_{\ell=N}\leq \mathcal{E}_{\theta}|_{\ell=s} \]
for any $s<N$. 
\end{proposition}
This means that if we fix the frequency of the inflow $\theta$ and the number of vertices $N$, the worst situation of the boundary for the comfortability is that every vertex is the boundary. 
\begin{proof}
Let us set $\mathcal{E}(\theta,m)$ as the comfortability at the frequency $\theta$ with the boundary $N-m$. 
We will show that $\mathcal{E}(\theta,m)-\mathcal{E}(\theta,0)>0$ for any $m\geq 1$. 
By (\ref{eq:comf2}), we have  
\begin{equation}
    \mathcal{E}(\theta,m)-\mathcal{E}(\theta,0)=\frac{N-1}{N(N-2)(\alpha^2+1)-2\alpha}\left( \frac{\alpha^2f_1+m(2\alpha+m+1)f_2}{\alpha^2f_3-m(\alpha-(N-2))(2\alpha+m)} -\frac{f_1}{f_3}\right)
\end{equation}
Let us detect the signature of ``$(\;\;)$" in the RHS. By reducing the common denominator, it is enough to show that \[g_m(x):=-(2x+m+1)h(x)+f_1(x;N)>0\;\;(m=1,2,\dots,N-1)\] 
for any $x\in[-N,N-2]$, where $h(x)$ is the following quadratic function of $x$ such that
\begin{equation*}
    h(x):=Nx^2-N(N-3)x-(N^3-2N^2+N-1).
\end{equation*}
Since $h(-N)>0$ while $h(0)<0$, the zero for $h(\gamma)=0$ uniquely exists between $-N$ and $0$.
Thus if $x\in[\gamma,N-2]$, then $h(x)\leq 0$. 
Then if $x\in[\gamma,N-2]$, we have $g_m(x)=(2x+m-1)|h(x)|$, which  
implies that $g_m(x)\geq g_1(x)$ for any $x\in[\gamma,N-2]$. 
On the other hand, let us consider the case for  $x\in[-N,\gamma]$, where $h(x)\geq0$. 
Note that such $\gamma$ can be bounded from the above by
\[\gamma=\frac{N(N-3)-(N-1)\sqrt{N(5N-4)}}{2N}<\frac{N(N-3)-(N-1)\sqrt{4N^2}}{2N}=-\frac{N+1}{2} \]
for any $N\geq 4$. 
If $x\in[-N,\gamma]$ with $N\geq 4$, since 
\begin{align*}
    2x+m+1 &< 2\gamma+m+1<2\left(-\frac{N+1}{2}\right)+m+1=m-N<0, 
\end{align*}
we confirm $g_m(x)=|2x+m+1| h(x)$. 
This implies that 
\[g_m(x)\geq g_1(x) \text{ for any $m=2,3,\dots,N-1$.}\] 
The positivity of $g_m(x)$ in the case for $N=3$ can be checked directly.  

Now the rest of our task is to show $g_1(x)>0$ for $x\in[-N,N-2]$ with $N\geq 4$. 
We have
\begin{multline*}
    g_1(x)=-2N(N-2)x^3+2N(N-2)(N-4)x^2\\+(3N^4-12N^3+11N^2-2N+6)x+(N-2)(N^2+1)(N^2-N-1).
\end{multline*} 
Let $\gamma_0,\gamma_\pm$ be the three solutions of the cubic equation $g_1(x)=0$. 
Comparing each coefficient in the above expression with that of $g_1(x)=-2N(N-2)(x-\gamma_0)(x-\gamma_+)(x-\gamma_-)$, we have   $\gamma_0+\gamma_++\gamma_-=N-4$ and $\gamma_0\gamma_+\gamma_-=(N^2+1)(N^2-N+-1)/(2N)$. 
Thus $g_1(x)$ can be reexpressed by 
\[ g_1(x)=-2N(N-2)(x-\gamma_0)\left(x^2-(N-4-\gamma_0)x+\frac{(N^2+1)(N^2-N-1)}{2N\gamma_0}\right). \]
It is easy to check that $g_1(N-2)>0$ while $g_1(2(N-2))<0$.
Then there exists a real-valued zero of $g_1(x)$ in $(N-2,\;2(N-2))$ which is out of the range $[-N,N-2]$. Next, our task is to show that $\gamma_0$ is the unique zero of $g_1(x)$ in $\mathbb{R}$; in other words,  
\begin{enumerate}
    \item $N-2<\gamma_0<2(N-2)$;
    \item The discriminant of the quadratic equation $g_1(x)/(-2N(N-2)(x-\gamma_0))=0$ is negative.
\end{enumerate}
The discriminant of the quadratic equation``$(\;)"=0$ in $g_1(x)$ can be expressed by 
$\tau(\gamma_0)/(N\gamma_0)$, 
where 
\[ \tau(y)=Ny^3-2N(N-4)y^2+N(N-4)^2y-2(N^2+1)(N^2-N+1). \]
Then it is enough to show that $\tau(y)<0$ for any $y\in[N-2,\;2(N-2)]$. 
The derivative of $\tau(y)$ can be simply described by
\[ \tau'(y)=N\left(\;3y-(N-4)\;\right)\;\left(\;y-(N-4)\;\right). \]
Thus $\tau(y)$ takes the local maximum and minimum values at $y=(N-4)/3$ and $N-4$, respectively, if $N\geq 5$ while $\tau(y)$ is a monotone increasing function when $N=4$. 
The value at the boundary $y=2(N-2)$ is negative; that is, $\tau(2(N-2))<0$ for $N\geq 4$. 
If $N\geq 5$, the local maximal value is negative; that is, $\tau((N-4)/3)<0$. This implies $\tau(y)<0$ for any $y\in(N-2,\;2(N-2))$ with $N\geq 4$. Therefore $g_1(x)$ has the unique real-valued zero $\gamma_0$ in $(N-2,2(N-2))$. 
\end{proof}

\noindent{{\bf (4) The most uncomfortable frequency is $\theta=\pi$ for $N\geq 4$. }} \\
We will show that if we fix the size $N\geq 4$ and the number of boundaries $\ell$, the most uncomfortable frequency for the quantum  walker is $\theta=\pi$.
Note that the comfortability at the frequency $\theta=\pi$ is independent of $m$. 
Since $\mathcal{E}(\theta,m)\geq \mathcal{E}(\theta,0)$ by Proposition~\ref{prop:Kenergy}, it is enough to show that $\mathcal{E}(\pi,0)$ takes the minimum value only at $\theta=\pi$. 
Let us compute the difference  
\begin{equation}
    \mathcal{E}(\theta,0)-\mathcal{E}(\pi,0)=\frac{(N-1)f_1}{(N(N-2)(\alpha^2+1)-2\alpha)f_3}-\frac{2N-3}{2(N-1)(N-2)}.
\end{equation}
By reducing the common denominator, it is enough to show that
\begin{multline*}
    w(x)=2(N-1)^4(N-2)\left(\; (N^2-4N+2)x+(N-2)(N^2-N+1) \;\right)\\
    -(2N-3)\left(\; N(N-2)(x^2+1)-2x \;\right)\left(\; N(N-2)x+N^3-2N^2+2N-2 \;\right)>0
\end{multline*} 
for any $x\in[-N,N-2]$. 
On the other hand, since $\theta=\pi$ is equivalent to $\alpha=N-2$, $w(x)$ can be divided by $(N-2-x)$. This means that $w(x)$ can be described by 
\[ w(x)=(N-2-x)(a_2x^2+a_1xx+a_0). \]
Comparing with each coefficient, we obtain the values of $a_j$; for example, $a_2=(2N-3)N^2(N-2)^2>0$. 
The discriminant of $a_1x^2+a_1x+a_0$ can be estimated by
\begin{align*}
    D/4 := (a_1/2)^2-a_2a_0<0.
\end{align*}
This implies $a_2x^2+a_1x+a_0>0$; that is, $w(x)>0$ for any $x\in[-N,(N-2)]$. Then we can state that the worst frequency for the comfortability is $\theta=\pi$. 
The minimum comfortabilities for $N=3$ are   
$\min_\theta \mathcal{E}(\theta,0), \min_\theta \mathcal{E}(\theta,1)<3/4$, and $\min_\theta \mathcal{E}(\theta,2)=3/4$. 

Since $\mathcal{E}(\theta;0)$ is the bottom of all the functions of $\mathcal{E}(\theta,m)$, 
it would be worthwhile to obtain an abstract shape of this bottom function. 
Instead of $\theta$, we choose $\alpha$ as the parameter of $\mathcal{E}(\theta;0)$. Let us set $\mathcal{E}_*(x):=\mathcal{E}(\theta;0)$ with $x=-1-(N-1)\cos\theta$, $\alpha\in[-N,N-2]$. 
To find when the bottom function takes the local minimum and maximum values, we take the derivative of $\mathcal{E}_*(x)$ and its signature. It is enough to estimate the following signature of the function:
\[ v(x)=-(b_3x^3+b_2x^2+b_1x+b_0). \]
The coefficients $b_j$' are $b_3>0$ if $N\geq 4$; $b_3<0$ if $N=3$, and $b_2,b_1>0$, $b_0<0$. 
It is easy to see that $v(0)<0$ while $v(N-2),v(-N)>0$ which implies that $v(x)$ has the zeros one by one in $(-N,0)$ and $(0,N-2)$, respectively. Then if $N\geq 4$, $\mathcal{E}_*(x)$ uniquely takes the local minimum and maximum values in $(-N,0)$ and $(0,N-2)$, respectively, and the minimum value is located at $x=N-2$.  
On the other hand, if $N=3$, then $\min_{\theta}\mathcal{E}_*(x)$ coincides with the local minimum. \\

The dependence on the frequency of the input of the  comfortability is shown in Fig.~\ref{Fig:connection}. 
We have shown that the curve for $\ell=N$ is always the bottom of any other $\ell$'s, and the bottom curve always has three peaks at the origin and around $\theta=\pm \pi/2$, respectively.  

\begin{figure}[hbtp]
    \centering
    \includegraphics[keepaspectratio, width=100mm]{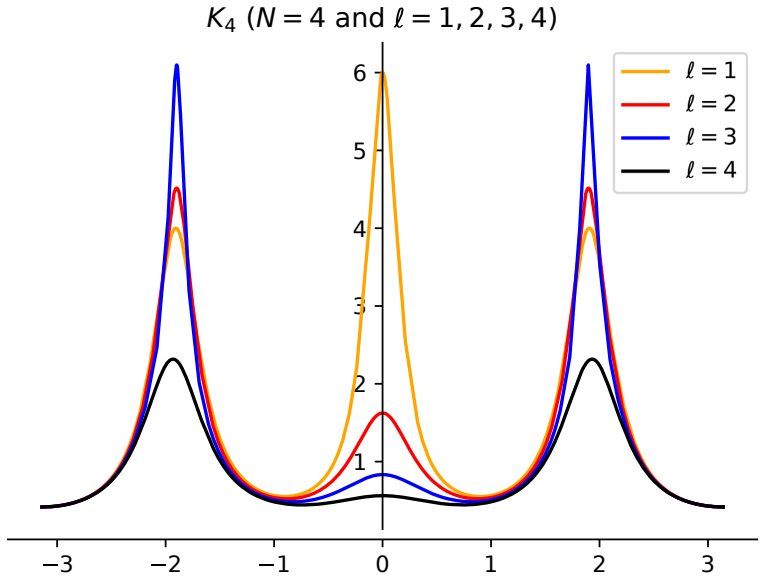}
    \caption{The comfortability of the quantum walk on the complete graph with $4$ vertices for each case for $\ell=1,2,3,4$:  
    The situation of the quantum walk for each $\ell$ is depicted in Figure~\ref{Fig:graphs}. The horizontal and vertical axes are the frequency of the inflow and the comfortability in the case of $N=4$. 
    The worst setting for quantum walkers of the boundaries is $\ell=N$ (black curve); that is, the joining of the tail to every vertex. 
    If the frequency is $\theta=0$, then the greater the number of boundaries is, the lager the  comfortability will be, which matches the classical intuition. On the other hand, if the frequency is around $\theta=\pm\theta_*$, the greater the number of boundaries is, the larger the comfortability will be, but once the number of boundaries is the maximum; that is, $\ell=N$, then the comfortability declines to its lowest level. }
    \label{Fig:connection}
\end{figure}

\noindent{{\bf (5) The quantum walker on the complete graph has a spiky preference for the frequency $\theta$ for large $N$}}\\
Next, let us examine the asymptotics of the comforatibility for large size $N$. 
From (\ref{eq:comf2}), we can estimate the comfortaibility for large size $N$ as follows. 
\begin{equation}\label{eq:comfasymp}
    \mathcal{E}(\theta) \sim \begin{cases} N^2/(2\ell^2) & \text{: $\theta=0$,} \\
    N/2 & \text{: $\theta=\pm \pi/2$,} \\
    1/(N\cos^2\theta) & \text{: otherwise.}
    \end{cases}
\end{equation}
Here $a_N\sim b_N$ is equivalent to $\lim_{N\to\infty}b_N/a_N=1$. 
Therefore the quantum walker on the complete graph has the ``spiky preference" of the frequency of inflow around $\theta=0$ and $\theta=\pm\pi/2$. 
Let us see the comfortability around $\theta=\pi/2$, because if $\theta$ is closed to $\pi/2$, then $\mathcal{E}(\theta)$ goes to infinity. 
This means that we will chase the asymptotics of $\mathcal{E}_{\theta}$ so that the parameter $\theta$ is tuned closer to $\pi/2$ as $N$ becomes larger.  
To this end, we set $\cos \theta=-s/(N-1)$ with $s\in o(N)$. Here $b_N\in o(a_N)$ is $\lim_{N\to\infty}b_N/a_N=0$. 
It is shown that the comfortability $\mathcal{E}(\theta)$ for any $\theta$ satisfying the above setting can be bounded above by $O(N)$ from the expression of (\ref{eq:comf2}). Here $b_N\in O(a_N)$ means $\lim_{N\to\infty} |b_N/a_N|<\infty$. Indeed we have
\begin{equation}\label{eq:comf3}
\mathcal{E}(\theta) \sim 
\begin{cases}
\frac{m+1}{m}N & \text{: $s=1$, $m=N-\ell\in \Theta(1)$,}\\
N & \text{: ``$s=1$, $m=0$" or ``$s\in \Theta(1)\setminus\{1\}$"}\\
N/(-1+s)^2 & \text{: otherwise}
\end{cases}
\end{equation}
for any $s\in o(N)$. 
Here $b_N\in\Theta(a_N)$ means $\lim_{N\to\infty}|b_N/a_N|\in(0,\infty)$.
Then for large $N$, if we increase the number of the boundaries, then 
the frequency of the maximal comfortability is switched from $\theta=0$ to $\theta=\theta_*$ when the number of boundaries is $\ell \in O(N^{1/2})$. 
Thus if $\ell\in \Omega(N^{1/2})$, then $\mathcal{E}(0)<\mathcal{E}(\theta_*)$ for sufficiently large $N$. 
On the other hand, without any asymptotics of $N$, the switching of the magnitude relation between $\mathcal{E}(0)$ and $\mathcal{E}(\theta_*)$ can be estimated as follows. By setting $m=t(N-1)$ with parameter $t\in[0,1]$ in (\ref{eq:comf2}), we can find that if $m<7(N-1)/8$, then  
$\mathcal{E}(0)<\mathcal{E}(\theta_*)$ for any $N\geq 3$. In any case, we observe that the situation of the boundary where  $\mathcal{E}(0)>\mathcal{E}(\theta_*)$ is rare. 

\noindent\\
\noindent {\bf Acknowledgments}
Yu.H. acknowledges financial supports from the Grant-in-Aid of
Scientific Research (C) Japan Society for the Promotion of Science (Grant No.~18K03401). 
E.S. acknowledges financial supports from the Grant-in-Aid of
Scientific Research (C) Japan Society for the Promotion of Science (Grant No.~19K03616) and by fund from Research Origin for Dressed Photon.



\begin{small}
\bibliographystyle{jplain}

\end{small}

\end{document}